\UseRawInputEncoding
\documentclass[12pt]{article}
\usepackage{amsmath}
\usepackage{graphicx,psfrag,epsf}
\usepackage{enumerate}
\usepackage{natbib}
\usepackage[noend]{algpseudocode}
\usepackage{algorithmicx, algorithm}
\usepackage{url} 
\usepackage{graphics,epsfig,epsf,psfrag}
\usepackage{url}
\usepackage{bm}
\usepackage{color, appendix}
\usepackage{graphicx}
\usepackage{amsmath, amsfonts, amsthm,amssymb}
\usepackage{epstopdf}
\usepackage{hyperref}
\usepackage{enumerate, framed}
\usepackage{graphicx,lscape,rotate}
\usepackage[noend]{algpseudocode}
\usepackage{algorithmicx, algorithm}
\usepackage{comment}
\usepackage{setspace}
\usepackage{booktabs}
\usepackage{multirow}

\newcommand{\blind}{0}

\numberwithin{equation}{section}
\newtheorem{theorem}{Theorem}[section]
\newtheorem{lemma}{Lemma}[section]

\newtheorem{rem}{Remark}[section]
\newtheorem{col}{Corollary}[section]

\newcommand{\pn}{\mathbb{P}_n}
\newcommand{\gn}{\mathbb{G}_n}
\newcommand{\E}{\mathrm{E}}
\newcommand{\pr}{\mathrm{P}}
\newcommand{\var}{\operatorname{var}}
\newcommand{\argmin}{\mathop{\operatorname{argmin}}_}

\newcommand{\X}{\mathrm{X}}
\newcommand{\spt}{\operatorname{spt}}

\begin{document}

\def\spacingset#1{\renewcommand{\baselinestretch}%
{#1}\small\normalsize} \spacingset{1}


\if0\blind
{
  \title{\bf Heterogeneous Overdispersed Count Data Regressions via Double Penalized Estimations}
  \author{\textbf{Shaomin Li}\footnote{Center for Statistics and Data Science, Beijing Normal University at Zhuhai. Email: \texttt{lsmjim@bnu.edu.cn} (Shaomin Li). This work is supported in part by National Natural Science Foundation of China Grants 12101056.}, \quad \textbf{Haoyu Wei}\footnote{Corresponding Author. Guanghua School of Management, Peking University. Email: \texttt{cute@pku.edu.cn} (Haoyu Wei). \quad The first two authors are co-first authors, in alphabetical order.}, \quad \textbf{Xiaoyu Lei}\footnote{Department of Statistics, University of Chicago. Email: \texttt{leixy@uchicago.edu} (Xiaoyu Lei)}
}
  \maketitle
} \fi

\if1\blind
{
  \bigskip
  \bigskip
  \bigskip
  \begin{center}
    {\LARGE\bf Title}
\end{center}
  \medskip
} \fi

\bigskip
\begin{abstract}
This paper studies the non-asymptotic merits of the double $\ell_1$-penalty for heterogeneous overdispersed count data via negative binomial regressions. Under the restricted eigenvalue conditions, we prove the oracle inequalities for Lasso estimators of two partial regression coefficients for the first time, using concentration inequalities of empirical processes. Furthermore, derived from the oracle inequalities, the consistency and convergence rate for the estimators are the theoretical guarantees for further statistical inference. Finally, both simulations and a real data analysis demonstrate that the new methods are effective.
\end{abstract}

\noindent%
\vfill

\newpage
\spacingset{1.8} 
\section{Introduction}

With the advance of modern data collection techniques, scientists and engineers generate or get access to a massive number of variables in their experiments, and challenges to traditional statistical methods and theories have come. One crucial aspect is the high-dimensional setting, where the number of covariates can be comparable to or greater than the sample size. In high-dimensional settings, the asymptotical results for the estimator are intractable. For example, the maximum likelihood estimator (MLE) in classical multivariate statistics is easy to obtain since the Hessian matrices in optimizations are invertible, and classical Newton algorithms also perform well. When the dimension is high, however, the optimizations for the MLE method lead to many un-meaningful solutions, so the specific constraint for variables in the optimizations is essential. The Lasso-regularization method is one of the constrained least-squares methods widely applied in the high-dimensional parameter estimation problem; see \cite{1996Regression}.

In many scientific fields such as biomedical science, ecology, and economics, experimental and observational studies often yield count data, a type of data in which the observations can take only the non-negative integer values.
The Poisson regression models are commonly used for count data. However, it needs a restrictive assumption that the variance equals the mean.  For many count data,  the variance is often larger than the mean \citep{dai2013maximum,zhang2018negative}, which is called over-dispersion.
Since the Poisson regression model is invalid under the over-dispersion case, a more general and flexible regression model, the negative binomial regression (NBR), has attracted lots of research attention and become popular in analyzing count data \citep{xie2020consistency}. 
Recently, there has been much research on the high-dimensional NBR model, such as  \cite{Qiu17, weissbach2020consistency, tian2020seasonal}. However, these works heavily rely on the distributional assumption of the count data with the pre-specified dispersion parameter. 

Most studies on NBR assumed the dispersion parameter as a constant. In practice, however, not all models satisfy the assumption. 
Thus the need to model the dispersion parameter as a function of some covariates.
The heterogeneous negative binomial regression (HNBR) extends the NBR by observation-specific parameterization of the dispersion parameter~\citep{Hilbe11}.
HNBR is a valuable tool for assessing the source of overdispersion. 
It belongs to the double generalized linear models (DGLMs) or vector generalized linear models (VGLMs), which are very useful in fitting more complex and potentially realistic models  \citep{yee2015vector,nguelifack2019robust}.
However, little work has been done to select the dispersion explanation variables.

In this paper, we study the variable selection and dispersion estimation for the heterogeneous NBR models. To the best of our knowledge and based on the literature, this study is the first.  
Specifically, we propose a double regression to estimate the coefficients of NB dispersion and NBR simultaneously. 
Because of the high dimension of the covariates, we apply a double $\ell_1$ penalty to both regressions.   
The two adjustment parameters we set are different because the first-order conditions for estimating the regression coefficients are entirely different from those for estimating the dispersion parameters.
We construct an algorithm to do variable selection and dispersion estimation simultaneously.
Similar studies on high-dimensional NBR models include \cite{wang2016penalized}, which assumed the dispersion parameter as a constant. 
Their method requires an iterative algorithm to estimate the mean regression and dispersion alternatively and implement lasso in each iteration.
If there are many iterations, such an algorithm is a waste of computing resources.

The rest of the paper is organized as follows. Section 2 introduces the heterogeneous overdispersed count data model and defines the double $\ell_1$-penalized estimators for the mean and dispersion regressions. Then we use a technique called the stochastic Lipschitz condition to derive the asymptotic results in Section 3.  Simulation studies and a real data application are given in Section 4. Finally, section 5 concludes the article with a discussion. All proofs and technical details are provided in the Appendix.

\section{Double $\ell_1$-penalized NBR}
\label{sec-asymptotics}

\subsection{Heterogeneous overdispersed count data regressions}

Suppose we have $n$ count responses ${Y_i}$ and  $p$-dimensional covariates $X_i = ({x_{i1}}, \cdots ,{x_{ip}})$, $i \in [n] := \{1, 2, \ldots, n \}$. For the Poisson regression models, the response obeys the Poisson distribution
\begin{equation*}
    \pr ({Y_i=y_i}{\left| \lambda  \right._i}) = \frac{{\lambda _i^{{y_i}}}}{{{y_i}!}}{e^{ - {\lambda _i}}}, \qquad i \in [n].
\end{equation*}
With ${\lambda _i} = {\rm{E}}{(Y_i)}$, we require that the positive parameter ${\lambda _i}$ is related to a linear combination of $p$ covariates. A plausible assumption for the link function is $\eta ({\lambda _i}) = \log (\lambda _i) = {X_i^{\top}}\beta$.
It is worth noting that
$\mathrm{E}(Y_i| X_i)=\mathrm{var}(Y_i|X_i)=\exp({ X_i^{\top}}\beta)>0.$

For the traditional negative binomial regression, it assumes that the count data response obeys the NB distribution with over-dispersion:
\begin{equation}\label{eq1}
    \pr({Y_{\rm{i}}} = {y_i}|{X}_i) = :f({y_i}; k ,{\mu _i}) = \frac{{\Gamma (k  + {y_i})}}{{\Gamma (k ){y_i}!}}{(\frac{{{\mu _i}}}{{k  + {\mu _i}}})^{{y_i}}}{(\frac{k }{{k  + {\mu _i}}})^k }, \qquad i \in [n],
\end{equation}
with $\E ({Y_i} | X_i ) = {\mu _i}=\exp(\beta^{\top} X_i)$ and $k$ is an unknown qualification of the overdispersion level.
When $k  \to \infty $, we have ${\rm{var}}({Y_i} | X_i) = {\mu _i} + \frac{{\mu _i^2}}{k } \to {\mu _i}{\rm{ = E(}}{Y_i} | X_i )$, the Poisson regression for the mean parameter ${\mu _i}$.
Thus the Poisson regression is a limiting case of negative binomial regression when the dispersion parameter $k$ tends to infinite. 

In the heterogeneous negative binomial regression, $k$ is proposed a specific parameterization, i.e. $k = k(X_i)$. More specifically,  we assume in this paper that
\begin{equation*}
    \mu(x) = \exp\{ \theta^{(1) \top} x\}, \qquad k(x) = \exp\{ \theta^{(2) \top} x \}.
\end{equation*}
For notation simplicity, we denote
\begin{equation*}
    P f := \E f(X_i, Y_i), \qquad \pn f := \frac{1}{n} \sum_{i = 1}^n f(X_i, Y_i), \qquad \gn f := \sqrt{n} (\pn - P)f,
\end{equation*}
for any measurable function $f$. 
\vspace{2ex}

Let $\theta=(\theta^{(1)\top},\theta^{(2)\top})^{\top}\in \mathbb{R}^{2p}$,the log-likelihood is
\begin{align*}
    n \ell (\theta) & = \log \prod \limits_{i = 1}^n {f({y_i},k_i ,{\mu _i})} = \sum\limits_{i = 1}^n {\log } \left\{ \frac{{\Gamma (k_i  + {Y_i})}}{{\Gamma (k_i ){Y_i}!}}{(\frac{{{\mu _i}}}{{k_i  + {\mu _i}}})^{{Y_i}}}{(\frac{k_i }{{k_i  + {\mu _i}}})^{k_i} }\right\} \\
    &=  \sum\limits_{i = 1}^n \bigg[-\log \frac{\Gamma \big(\exp \{ X_i^{\top} \theta^{(2)} \}\big)}{\Gamma \big( {Y_i} + \exp\{ X_i^{\top} \theta^{(2)}\}\big)} + {Y_i} X_i^\top\big(\theta^{(1)} - \theta^{(2)} \big) \\
    & \qquad \qquad \qquad \qquad \qquad \qquad - \big[ Y_i+ \exp \{ X_i^{\top} \theta^{(2)} \}\big] \log \Big( 1 + \exp \{ X_i^{\top} \big(\theta^{(1)} - \theta^{(2)} \big)\}\Big) - \log {Y_i}!\bigg]
\end{align*}
We use the negative log-likelihood as the loss function $\gamma$, and define
\begin{equation*}
    \begin{aligned}
        \gamma(\theta) & := - \log f(y \, | \, x, \theta) + \log y!.\\
    \end{aligned}
\end{equation*}

Denote $\partial_j := \frac{\partial}{\partial \theta^{(j)}}$, $j = 1, 2$, the score function for $\theta^{(1)}$ is
\begin{equation*}
    \partial_1 \ell (\theta) =  \mathbb{P}_n \partial_1 \gamma (\theta) = \frac{1}{n} \sum\limits_{i = 1}^n {({Y_i} - {{\mathop{ e}\nolimits} ^{X_i^ \top {\theta ^{(1)}}}})\frac{{{{\mathop{ e}\nolimits} ^{X_i^ \top {\theta ^{(2)}}}}{X_i}}}{{{{\mathop{ e}\nolimits} ^{X_i^ \top {\theta ^{(1)}}}} + {{\mathop{ e}\nolimits} ^{X_i^ \top {\theta ^{(2)}}}}}}}.
\end{equation*}
Furthermore, fix $\theta^{(1)}$, the score function for $\theta^{(2)}$ is
\[  \partial_2 \ell (\theta) =  \mathbb{P}_n \partial_2 \gamma (\theta) = \frac{1}{n} \sum\limits_{i = 1}^n {\left\{ {\left[ {\log \left( {1 + {{\mathop{ e}\nolimits} ^{X_i^ \top ({\theta ^{(1)}} - {\theta ^{(2)}})}}} \right) - \sum\limits_{j = 0}^{{Y_i} - 1} {\frac{1}{{j + {{\mathop{ e}\nolimits} ^{X_i^ \top {\theta ^{(2)}}}} }}} } \right]{\rm{ + }}\frac{{{Y_i} - {{\mathop{ e}\nolimits} ^{X_i^ \top {\theta ^{(1)}}}}}}{{{{\mathop{ e}\nolimits} ^{X_i^ \top {\theta ^{(1)}}}} + {{\mathop{ e}\nolimits} ^{X_i^ \top {\theta ^{(2)}}}}}}} \right\}{{\mathop{ e}\nolimits} ^{X_i^ \top {\theta ^{(2)}}}}{X_i}} .\]
It is easy to verify that 
\begin{equation*}
    P \partial_1 \ell(\theta) = P \partial_2 \ell (\theta) = 0.
\end{equation*}
So from now, we will suppose the true parameter is $\theta^*$. 

\subsection{Heterogeneous overdispersed NBR via double $\ell_1$-penalty}

The lasso estimator under our circumstance is defined as
\begin{equation}\label{lasso}
    \widehat{\theta} = \argmin{\theta \in \Theta} \big( \pn \gamma(\theta) + \lambda \| \theta\|_{\omega, 1} \big),
\end{equation}
where $\lambda > 0$ is the tuning parameter and the weighted norm is defined by 
\begin{equation*}
    \lambda \| \theta \|_{\omega, 1} =  \lambda_1 \| \theta^{(1)}\|_{1} + \lambda_2 \| \theta^{(2)} \|_{1} = \lambda \left( \omega_1 \| \theta^{(1)}\|_1 + \omega_2 \| \theta^{(2)}\|_1 \right),
\end{equation*}
and $\omega = (\omega_1, \omega_2)^{\top} = (\lambda_1 / \lambda, \lambda_2 / \lambda )^{\top} \in [0, 1] \times [0, 1]$ is the weight, 
$\|\cdot \|_1$ means the $\ell_1$-norm.
This technique is also used in \cite{huang2021weighted}.
Equation (\ref{lasso}) is a weighted double $\ell_1$-penalized problem, which is a kind of convex penalty optimization, and when $\lambda_1=\lambda_2$, it becomes a single penalized problem. 
In this paper, we use different $\lambda_1$ and $\lambda_2$, since the first-order conditions for estimating the regression coefficients are entirely different from those for estimating the dispersion parameters, and take $\lambda = \lambda_1 \vee \lambda_2$.

Since the weighted group lasso estimator $\widehat{\theta}_{n}$ has no closed-form solution, we need to use iterative methods like quasi-Newton or coordinate descent methods. We use BIC to choose the parameter $\lambda_1$ and $\lambda_2$. 
\begin{equation*}
    \operatorname{BIC}(\lambda_1, \lambda_2)=-2\ell(\widehat{\theta}_n)+ \dfrac{\log n}{n}k,
\end{equation*}
where $k$ is the number of nonzero estimated coefficients. To illustrate the algorithm explicitly, we rewrite $\gamma(\theta)$ as $\gamma(\theta^{(1)\top}x,\  \theta^{(2)\top}x)$, and define $\theta^{(3)} =\lambda_2/ \lambda_1\theta^{(2)}$, $\theta^{\dag}=(\theta^{(1)\top},\theta^{(3)\top})^{\top}$. Converting $\theta^{(2)}$ into $\theta^{(3)}$ turns the double $\ell_1$-penalized problem into a single penalized one, which is very to solve through some R packages. The algorithm is formally given in Algorithm 1.
\begin{algorithm}[t]
\caption{Double $\ell_1$-Penalized Optimization} 
\hspace*{0.02in} {\bf Input:} 
the set of tuning parameters  $\Lambda=\{(\lambda_{1,i}, \lambda_{2,i}) \}_{i=1}^m$\\
\hspace*{0.02in} {\bf Output:} 
the estimate $\widehat{\theta}_n$
\begin{algorithmic}
\For{$i=1,...,m$, } 
　　\State  let $x^*= \frac{\lambda_{1,i}}{\lambda_{2,i}}x$;
　　\State solve   $\widehat{\theta}^{\dag} = (\widehat{\theta}^{(1)\top}, \widehat{\theta}^{(3)\top})^{\top}=\argmin{\theta^{\dag} \in \Theta} \big( \pn \gamma(\theta^{(1)\top}x,\  \theta^{(3)\top}x^*)) + \lambda_{1,i} \| \theta^{\dag}\|_{1} \big)$;
　　\State obtain the estimate $\widehat{\theta}_{n,i}=(\widehat{\theta}^{(1)\top},  \frac{\lambda_{2,i}}{\lambda_{1,i}}\widehat{\theta}^{(3)\top})^{\top}$;
　　\State compute $ \operatorname{BIC}(\lambda_{1,i}, \lambda_{2,i})=-2 \ell (\widehat{\theta}_{n,i})+ \dfrac{\log n}{n}k_i$;
\EndFor
\State {\bf end for}
\State find $i_{opt}=\argmin{i=1,...,m} {\operatorname{BIC}(\lambda_{1,i}, \lambda_{2,i})}$;
\State \Return $\widehat{\theta}_{n,i_{opt}}$
\end{algorithmic}
\end{algorithm}

The proposed algorithm can do variable selection and dispersion estimation simultaneously. Similar studies on high-dimensional NBR models include \cite{wang2016penalized}, which assumed the dispersion parameter as a constant. However, their method requires an iterative algorithm to estimate the mean regression and dispersion alternatively and implement lasso in each iteration. If there are many iterations, such an algorithm is a waste of computing resources.

\section{Main results}
\subsection{Stochastic Lipschitz conditions}

We write the maximum of $Y_i$ from the sample of size $n$ as $M_{Y, n}$, then the sample space for $\{Y_i\}_{i = 1}^n$ is $\mathcal{Y} := \{y \in \mathbb{N}, \, y \leq M_{y, n}\}$. i.e. $M_{y, n} = \max_{i \in [n]} Y_i$. Note that $\lim_{n \rightarrow \infty} \pr ( M_{y, n} = \infty ) = 1$, what we need to tackle is actually a unbounded empirical process.
But for $ z := \left( \begin{matrix}
       x^{\top} & \\
         & x^{\top}
   \end{matrix} \right) \in \mathbb{R}^{2 \times 2p}$, we can assume the value space $\mathcal{S}$ for $s := z \theta$ is bounded and satisfies
\begin{equation*}
    \mathcal{S} := \big\{ s = (s_1, s_2)^{\top} \in \mathbb{R}^2, \ -\infty < m_{s, n} \leq s_j \leq |s_j| \leq M_{s, n} < \infty, \ j = 1, 2 \big\}.
\end{equation*}

As we can see, the most significant difference between this article and other conventional literature about lasso estimators is that we use $s = z \theta$ rather than $\theta$ as the explanatory variable to analyze the properties of the loss function $\gamma$. This is not a traditional way. At the first glimpse, the combination may complicate the analysis in the next step because the KKT condition requires the story about $\frac{\partial}{\partial \theta} \gamma$, which is critical for the traditional convex penalty problem. However, this article will try a different approach, the stochastic Lipschitz conditions introduced in the event $\mathcal{A}$ of Proposition 1 in \cite{zhang2022elastic}, to solve $\ell_1$-penalization problem. Define the \emph{local stochastic Lipschitz constant }by
\begin{equation*}
    \operatorname{Lip} (f; \theta^*) := \sup_{\theta \in \Theta / \{\theta^*\}} \left| \frac{\sqrt{n} \gn \big( f (\theta) - f (\theta^*)\big)}{\| \theta - \theta^*\|_1} \right|.
\end{equation*}
The most  apparent advantage of the stochastic Lipschitz conditions over KKT condition is that it can easily deal with the several parameters involved in different locations of the model that need to impose the same penalty on them, which is why we do not need to derive KKT condition in this paper.

To establish the stochastic Lipschitz conditions for this unbounded counting process, another assumption called \textit{strongly midpoint log-convex} for some positive $\gamma$ should be satisfied, which states for the joint density from the sample $\mathbb{Y} := (Y_1, \ldots, Y_n)^{\top} \in \mathbb{Z}^n$'s negative log-density of $n$ independent NB responses $\psi({y}):=-\log p_{\mathbb{Y}}({y})$ satisfies
\begin{equation*}
    \psi({x})+\psi({y})-\psi\left(\left\lceil\frac{1}{2} {x}+\frac{1}{2} {y}\right\rceil\right)-\psi\left(\left\lfloor\frac{1}{2} {x}+\frac{1}{2} {y}\right\rfloor\right) \geq \frac{\gamma}{4}\|{x}-{y}\|_{2}^{2}, \qquad \forall {x}, {y} \in \mathbb{Z}^{n}.
\end{equation*}
This assumption is a condition that ensures that the suprema of the multiplier empirical processes of $n$ independent responses have sub-exponential concentration phenomena, which can be alternatively be checked by the tail inequality for suprema of empirical processes corresponding to classes of unbounded functions (\cite{adamczak2008tail}). 

{
{
\begin{col}\label{thm1}
    Suppose $\max_{i \in [n], \, 1 \leq k \leq p} |X_{ik}| \leq M_x < \infty$, the parameter space $\Theta$ is convex and its diameter $D_{\Theta} < \infty$. If $\{Y_i\}_{i = 1}^n$ and $\{Z_i \theta\}_{i \in [n], \, \theta \in \Theta}$ are both in the value space $\mathcal{Y}$ and $\mathcal{S}$ defined as previous, then for any $\theta \in \Theta$,
    \begin{equation*}
         \begin{aligned}
             \operatorname{Lip}(\gamma; \theta^*) & = \sup_{\theta \in \Theta / \{\theta^*\}} \left| \frac{\sqrt{n} \gn \big( \gamma (\theta) - \gamma (\theta^*)\big)}{\| \theta - \theta^*\|_1} \right|  \\
             & \leq \sqrt{n} M_q:= \Big( A_1 \sqrt{\log (2p / q_2)} + A_2 \sqrt{\log p} + A_3 \sqrt{\log (p / q_3)}\Big) \sqrt{\max_{1 \leq k \leq p} \sum_{i = 1}^n X_{ik}^2 } \\
             & + B \sqrt{\log (2p / q_1)} \sqrt{\Big( \max_{1 \leq k \leq p} \sum_{i=1}^{n} X_{ik}^4 \Big)^{1/2}} \vee C \log (2p / q_1) + D \log (p / q_3),
         \end{aligned}
    \end{equation*}
    with probability at least $1 - q_0$, where $q_1, q_2, q_3 \in (0, 1)$ satisfy $q_1 + q_2 + q_3 = q_0$, and the constants are as follows:
    \begin{gather*}
        A_1 = \sqrt{2} F_1, \qquad A_2 = 32\sqrt{2} M_x F_2 D_{\Theta}, \qquad A_3 = \sqrt{2} \big( 2 (F_1 + M_{y, n}) \vee F_2 M_x D_{\Theta}\big), \\
        B = 6 \sqrt{2 \big( w^{(1)} \vee w^{(2)} \big) \bigg( \sum_{i = 1}^n a(\mu_i, k_i)^4\bigg)^{1/2}}, \qquad C = 12 M_x \big( w^{(1)} \vee w^{(2)} \big) \max_{1 \leq i \leq n} a(\mu_i, k_i), \\
        D = 8 \big( 2 (F_1 + M_{y, n}) \vee F_2 M_x D_{\Theta}\big) M_x, \qquad w^{(1)} = \frac{e^{M_{s, n}}}{e^{m_{s, n}} + e^{M_{s, n}}}, ~~ w^{(2)} = \frac{e + e^{M_{s, n} - m_{s, n}}}{1 + e^{m_{s, n} - M_{s, n}}} + \frac{1}{1 + e^{m_{s, n} - M_{s, n}}},
    \end{gather*}
    where $M_{y, n} = \max_{i \in [n]} Y_i$ is the suprema empirical process.
\end{col}
}}

{
{

It is worthy to note that the $M_{y,n}$ in Corollary \ref{thm1} is a random process, hence the bound above is not deterministic. Fortunately, $M_{y, n}$ can use the strongly midpoint log-convex condition to be bounded, which we state in Lemma \ref{log_concave.lem}. Corollary \ref{thm1} combined Lemma \ref{log_concave.lem} will give the following corollary as a step more. 

\begin{col}\label{col2}
    Assume the conditions are the same as that in Corollary \ref{thm1}, then the stochastic Lipschitz constant has a nonrandom upper bound:
    \begin{equation*}
         \begin{aligned}
             \operatorname{Lip}(\gamma; \theta^*) & \leq \sqrt{n} M_q^{\prime}:= \Big( A_1 \sqrt{\log (2p / q_2)} + A_2 \sqrt{\log p} + 2 A_3^{\prime} \big( \log (2n / q_4) + \sqrt{\log (np / q_3) \big)}\Big) \sqrt{\max_{1 \leq k \leq p} \sum_{i = 1}^n X_{ik}^2 } \\
             & + B \sqrt{\log (2p / q_1)} \sqrt{\Big( \max_{1 \leq k \leq p} \sum_{i=1}^{n} X_{ik}^4 \Big)^{1/2}} \vee C \log (2p / q_1) + D \log (p / q_3),
         \end{aligned}
    \end{equation*}
    with probability at least $1 - q_0$, where $q_1, q_2, q_3, q_4\in (0, 1)$ satisfy $q_1 + q_2 + q_3 + q_4 = q_0$, and
    $$A_3^{\prime} = 2 \sqrt{2}\left(F_1 +  \left( 2 \gamma \max_{i \in [n]} \left[ a(\mu_i, k_i) - \frac{\mu_i}{\log 2} \right] \right) \vee F_2 M_x D_{\Theta}\right).$$
\end{col}

}}

Theorem \ref{thm1} gives us a different sight of the loss function far more than KKT conditions. However, the stochastic Lipschitz condition above does not compare the estimated and true values directly. We can resolve this issue by using an eigenvalue condition on the design matrix consisting of $X_i$. Since the design matrix $\X$ is fixed, the eigenvalue condition in the next section is reasonable. It is worthy to note that this inequality is an oracle since it involves an unknown empirical process on the right side.

\subsection{$\ell_{2}$-estimation error oracle inequalities RE conditions}

As we said previously, although we use stochastic Lipschitz conditions instead of KKT conditions, the restricted eigenvalue conditions (RE conditions) are still required. We denote by $\delta_{J}$ the vector in $\mathbb{R}^{p}$ with the same coordinates as $v$ on $J$ and zero coordinates on the complement $J^{c}$ of $J$, and $\spt(v) = \{j : v_j \neq 0\}$. 
We will assume that the minima in (\ref{lasso}) can always be obtained in the following setting, but it may not be unique.
In general, to bound $\widehat{\theta} - \theta^*$, some conditions on the design matrix $\X \in \mathbb{R}^{n \times p}$ are needed for getting abound in terms of the $\ell_2$ norm of $\theta - \theta^*$. 
Here we will utilize the restricted eigenvalue condition introduced in \cite{bickel2009simultaneous}, which says that for some $1 \leq s \leq p$ and $K > 0$,
\begin{equation}\label{RE.condition}
    \kappa (s, K) = \min \left\{ \frac{\|\X v\|_2}{\sqrt{n} \|v_J\|_2} : 1 \leq |J| \leq s, \, v \in \mathbb{R}^p / \{0\}, \, \| v_{J^c} \|_1 \leq K \|v_J\|_1 \right\} > 0.
\end{equation}
It should be noted that omitting the weight $\omega$ and the sparse restricted set $\| v_{J^c} \|_1 \leq K \|v_J\|_1$ leads to $v^{\top} \big[ \frac{1}{n} \X^{\top} \X \big] v / v^{\top} v \geq \kappa^2 (s, K)$. Thus it means that the smallest eigenvalue of the sample covariance matrix $\frac{1}{n} \X^{\top} \X$ is positive, which is impossible when $p > n$  since $\frac{1}{n} \X^{\top} \X$ is not full rank. To avoid
this problem, \cite{bickel2009simultaneous} consider the restricted eigenvalue condition under the
sparse restricted set $\| v_{J^c} \|_1 \leq K \|v_J\|_1$ as considerable relation in sparse high-dimensional estimation. The restricted eigenvalue is from the restricted strong convexity, which enforces a strong convexity condition for the negative log-likelihood function of linear models under certain sparse restrict set.

Due to the double penalty, besides the RE condition, we also require another condition similar to RE condition so-called $l$-restricted isometry constant defined in \cite{candes2007dantzig} as follows
\begin{equation*}
    \sigma_{\X, l}^2 = \max \left\{ {\|\X v\|_2^2} / {\|v\|_2^2} : v \in \mathbb{R}^p, \, 1 \leq \spt(v) \leq l\right\} \in (0, \infty),
\end{equation*}
which essentially requires the eigenvalue of sample covariance matrix under every vector with cardinality less than $l$ ($l$ should be no more than $n$) approximately behaves normally like the low dimensional case.

With the RE condition and $l$-restricted isometry constant, and the two theorems we established before, the lasso estimator in (\ref{lasso}) can guarantee a good consistent property.

\begin{lemma}[see Lemma 3.1 in \cite{candes2007dantzig}]\label{ex.lem5}
    Suppose $T_{0}$ is a set of cardinality $S$. For a vector $h \in \mathbb{R}^{p}$, we let $T_{1}$ be the $S$ largest positions of $h$ outside of $T_0$. Put $T_{01}=T_{0} \cup T_{1}$, then
    \begin{equation*}
        \| h\|_2^2 \leq \|h_{T_{01}}\|_2^2 + S^{-1} \| h_{T_0^c}\|_1^2.
    \end{equation*}
\end{lemma}

\begin{theorem}\label{thm3}
    Suppose the condition is the same as that in Theorem \ref{thm2}. Furthermore, assume $p_1 = \spt \big( \theta^{(* 1)}\big) \vee \spt \big( \theta^{(* 2)}\big) \leq p / 2$, and there exists some $K > 1$, $\kappa : = \kappa (2p_1, K) > 0$. Let $\lambda = \frac{(K + 1)M_q}{n(K - 1)}$, then using this $\lambda$ in (\ref{lasso}), with probability at least $1 - q$,
    \begin{equation*}
        \| \widehat{\theta} - \theta^*\|_2^2 \leq \frac{8 p_1 M_q^{2\prime} K^2}{\kappa^4 n^2 C_{\gamma}^2 (K - 1)^2} \left[ 2 + K^2 + \frac{2(1 + 2 p_1 K^2)(n \kappa^2 + 2 \sigma_{\X, p_1}^2)}{n \kappa^2}\right],
    \end{equation*}
    where $M_q$, $C_{\gamma}$ are defined in Theorem \ref{thm1} and \ref{thm2} respectively.
\end{theorem}

\begin{rem}
    Compared to the single lasso problem, in which we only have one unknown vectorized parameter, the oracle inequality in Theorem \ref{thm3} has an extra term $\frac{2(1 + 2 p_1 K^2)(n \kappa^2 + 2 \sigma_{\X, p_1}^2)}{n \kappa^2}$.
\end{rem}

\begin{rem}
    From Theorem \ref{thm3}, we know that the $\ell_2$ convergence rate is minimax optimal, as studied in \cite{zhang2022elastic}.
\end{rem}

From Theorem \ref{thm3}, we can get the following corollary regarding the consistency of $\widehat{\theta}$ immediately by the property of the supreme empirical process.

\section{Numerical Studies}
\subsection{Simulations}
In this section, we evaluate the finite sample performance of the proposed method.
The response is generated from the  negative binomial regression model (\ref{eq1}) with
\begin{equation*}
    \mu(x) = \exp\{ \theta^{(1) \top} x\}, \text{and~} k(x) = \exp\{ \theta^{(2) \top} x \},
\end{equation*}
where $\theta^{(1)}$ and $\theta^{(2)}$ are two $p$-dimensional parameters.
The explanatory variables are generated from the multivariate normal distributions with mean vector $\bf{0}$ and $Cov(x_i,x_j)=\rho^{|i-j|}$, where $\rho=0, 0.5$.
The following two examples show the performance of the proposed estimator for the low dimensional heterogeneous negative binomial regression and the variable selection in the high dimensional case, respectively. The R package {\bf lbfgs} is required to solve the optimization problem.

{\bf Example 1 (Low dimension).} We set $p=3$ and $n=100,200,400$. The true parameters are  $\theta^{(1)}=(1,2,-1)$ and $\theta^{(2)}=(-1,0.5,1)$, and their maximum likelihood estimators are denoted as  $\hat{\theta}^{(1)}$ and $\hat{\theta}^{(2)}$, respectively.
We compare the  estimator $\hat{\theta}^{(1)}$ with  $\hat{\theta}^{(1)*}$, which ignores the heterogeneity of the overdispersion and treats  $k(x)$ as a constant.
Table~\ref{t1} displays the average squared estimation  errors $\| \hat{\theta}-{\theta} \|_2^2$ based on 200 repetitions.

We can make the following observations from the table.
Firstly, the performances of the three estimators become better and better as $n$ increases.
Secondly, the estimator $\hat{\theta}^{(1)}$, which estimates the parameter in the mean function $\mu(x)$, performs better than  $\hat{\theta}^{(2)}$, which estimates the parameter in the overdispersion function $k(x)$.
Last but the most important, $\hat{\theta}^{(1)*}$ performs much worse than $\hat{\theta}^{(1)}$.
For example, the average squared estimation error of  $\hat{\theta}^{(1)*}$ is about 5 times of $\hat{\theta}^{(1)}$'s when $n=100$, and 10 times of $\hat{\theta}^{(1)}$'s when $n=400$.
The comparison between $\hat{\theta}^{(1)}$ and $\hat{\theta}^{(1)*}$  indicates the necessity of considering the heterogeneity of the overdispersion.

\begin{table}[htbp]
\setlength{\belowcaptionskip}{1em}
	\centering
	\caption{ The average squared estimation errors of the estimators.}
	\label{t1}
	\begin{tabular}{ccccccc}
\hline
n&\multicolumn{3}{c}{$\rho=0$}&\multicolumn{3}{c}{$\rho=0.5$}\\
\cmidrule(lr){2-4}\cmidrule(lr){5-7}
&$\hat{\theta}^{(1)*}$&$\hat{\theta}^{(1)}$&$\hat{\theta}^{(2)}$&$\hat{\theta}^{(1)*}$&$\hat{\theta}^{(1)}$&$\hat{\theta}^{(2)}$\\
\hline
100	&0.1597	&0.0335	&0.72414	&0.1809	&0.0397	&0.68904	\\
200	&0.0862	&0.01	&0.22149	&0.0837	&0.0169	&0.33048	\\
400	&0.05	&0.0047	&0.08847	&0.0619	&0.0067	&0.15066	\\
\hline
	\end{tabular}
\end{table}

{\bf Example 2 (High dimension).}
The sample sizes are chosen to be $n=100, 200, 400$, with dimension $p\in (25,50,150)$, $(50,100,250)$ and $(100,200,500)$, respectively. We set $\theta^{(1)}=(1,2,-1,0,\dots,0)$ and $\theta^{(2)}=(-1,0.5,1,0,\dots,0)$. The unknown tuning parameters $(\lambda_1, \lambda_2)$  for the penalty functions are chosen by BIC criterion in the simulation. Results over 200 repetitions are reported. For each case, table~\ref{t2} reports the number of repetitions that each important explanatory variable is selected in the final model and also the average number of unimportant explanatory variables being selected.

It shows that the variable selection procedure performs better and better as the sample size $n$ increases. When $n=400$, the important explanatory variables in $\mu (x)$ and $k(x)$ are correctly selected in almost every repetitions.
When the dimension $p$ increases, the procedure may select more unimportant explanatory variables, but the average numbers are less than $1.3$. The important variables in $k(x)$ are less likely to be selected than the important variables in $\mu(x)$ especially when the sample size is small, as well as the unimportant variables.

\begin{table}[htbp]
\setlength{\belowcaptionskip}{1em}
	\centering
	\caption{ The results of variable selection.}
	\label{t2}
	\begin{tabular}{cccccccccc}
\hline
& &\multicolumn{4}{c}{$\mu (x)$}&\multicolumn{4}{c}{$k(x)$}\\
\cmidrule(lr){3-6}\cmidrule(lr){7-10}
n&p& $\theta^{(1)}_1$& $\theta^{(1)}_2$& $\theta^{(1)}_3$&Other $\theta^{(1)}$s
& $\theta^{(2)}_1$& $\theta^{(2)}_2$& $\theta^{(2)}_3$&Other $\theta^{(2)}$s\\
\hline
\multicolumn{10}{l}{$\rho=0$}\\
\hline
100&	25&	192&	200&	190&	0.37&	180&	184&	180&	0.32\\
&	50&	196&	200&	193&	0.52&	182&	180&	188&	0.41\\
&	150&	194&	194&	192&	1.02&	188&	182&	186&	0.41\\
\hline
200&	50&	200&	200&	200&	0.59&	200&	190&	198&	0.53\\
&	100&	200&	200&	200&	0.91&	196&	186&	196&	0.69\\
&	250&	199&	199&	198&	1.18&	198&	198&	198&	0.69\\
\hline
400&	100&	200&	200&	200&	0.4&	200&	198&	200&	0.55\\
&	200&	200&	200&	200&	0.6&	200&	200&	200&	0.51\\
&	500&	200&	200&	200&	1.21&	200&	200&	200&	0.61\\
\hline
\multicolumn{10}{l}{$\rho=0.5$}\\
\hline
100&	25&	194&	198&	194&	0.41&	179&	184&	180&	0.35\\
&	50&	196&	196&	190&	0.63&	178&	182&	180&	0.42\\
&	150&	194&	196&	192&	1.01&	180&	184&	182&	0.43\\
\hline
200&	50&	200&	200&	198&	0.38&	196&	183&	190&	0.32\\
&	100&	199&	200&	198&	0.53&	194&	186&	194&	0.44\\
&	250&	196&	198&	196&	1.1&	196&	196&	194&	0.55\\
\hline
400&	100&	200&	200&	200&	0.28&	200&	199&	194&	0.34\\
&	200&	200&	200&	198&	0.47&	200&	198&	196&	0.36\\
&	500&	200&	200&	198&	1.07&	200&	198&	196&	0.56\\
\hline
	\end{tabular}
\end{table}

\subsection{A real data example}
In this section, we apply the proposed method to the dataset of German health care demand. The data was employed in \cite{riphahn2003incentive} and could be downloaded on \url{http://qed.econ.queensu.ca/jae/2003-v18.4/riphahn-wambach-million/}.
The data contains 27,326 observations on 25 variables, including two dependent variables Docvis (number of doctor visits in last three months) and Hospvis (number of hospital visits in last calendar year). For conciseness, we focus on Docvis in this study.
We build the HNBR model based on the proposed variable selection procedure and make the standard NBR model a comparison. 
Define the fitting errors (FE) as  $n^{-1}\sum_{i=1}^n(y_i-\widehat{y}_i)$, where $y_i$ denotes the raw data of Hospvis, $\hat{y}_i$ is the predicted value, and $n$ is the sample size.
As the data is observed during 1984-1988, 1991, and 1994, we make the analysis each observed year.
Table~\ref{t-real} displays the variable selection results and fitting errors.

We have the following findings from the table. First, the important variables in the NBR are the same as HNBR models in each year, and the estimates are close. Second, the selected variables in $\mu (x)$ are almost the same every year, namely Age, Hsat (health satisfaction), Handper (degree of handicap), and Educ (years of schooling).
Moreover, some of these variables still play an essential role in $k(x)$, and $k(x)$ contains no variables other than these. Also we can see that the fitting errors of HNBR is less than that of NBR. 
All of these illustrate the advantage of our method. 
\begin{table}[htbp]\footnotesize
\setlength{\belowcaptionskip}{1em}
	\centering
	\caption{ The variable selection results and the fitting errors (FE) of NBR and HNBR models. The variable Others $=\{$ Married, Haupts, Reals, Fachhs, Abitur, Univ, Working, Bluec, Whitec, Self, Beamt, Public, Addon $\}$. Since these variables are not selected in any year, we put them in ``Others'' for brevity.}
	\label{t-real}
	\begin{tabular}{lcccccccccccc}
\hline
Variables&\multicolumn{3}{c}{$1984$}&\multicolumn{3}{c}{1985}&
\multicolumn{3}{c}{$1986$}&\multicolumn{3}{c}{1987}\\
\cmidrule(lr){2-4}\cmidrule(lr){5-7}
\cmidrule(lr){8-10}\cmidrule(lr){11-13}
&NBR&\multicolumn{2}{c}{HNBR}&NBR&\multicolumn{2}{c}{HNBR}
&NBR&\multicolumn{2}{c}{HNBR}&NBR&\multicolumn{2}{c}{HNBR}\\
\cmidrule(lr){3-4}\cmidrule(lr){6-7}
\cmidrule(lr){9-10}\cmidrule(lr){12-13}
&	&$\mu(x)$	&$k(x)$	&  &$\mu(x)$	& $k(x)$	&  & $\mu(x)$	& $k(x)$ &	&$\mu(x)$	&$k(x)$	\\
\hline
Female	&0	&0	&0	&0	&0	&0	&0	&0	&0	&0	&0	&0	\\
Age	&-0.013	&-0.013	&-0.012	&-0.009	&-0.01	&-0.007	&-0.006	&-0.006	&-0.013	&-0.002	&-0.001	&-0.018	\\
Hsat	&-0.205	&-0.2	&-0.025	&-0.244	&-0.237	&0	&-0.188	&-0.195	&-0.045	&-0.158	&-0.153	&-0.043	\\
Handdum	&0	&0	&0	&0	&0	&0	&0	&0	&0	&0	&0	&0	\\
Handper	&0.005	&0.005	&0.004	&0.007	&0.006	&0.007	&0.007	&0.007	&0	&0.007	&0.007	&0.01	\\
Hhninc	&0	&0	&0	&0	&0	&0	&0	&0	&0	&0	&0	&0	\\
Hhkids	&0	&0	&0	&0	&0	&0	&0	&0	&0	&0	&0	&0	\\
Educ	&0	&0	&-0.027	&0	&0	&-0.064	&-0.035	&-0.038	&0	&-0.095	&-0.106	&-0.003	\\
Others	&0	&0	&0	&0	&0	&0	&0	&0	&0	&0	&0	&0	\\
\hline
FE	&0.798	&\multicolumn{2}{c}{0.602}	&2.203	&\multicolumn{2}{c}{1.874}		&0.735	&\multicolumn{2}{c}{0.581}	&1.314	&\multicolumn{2}{c}{1.027}	\\
\hline
\\
\cmidrule(lr){1-10}
Variables&\multicolumn{3}{c}{$1988$}&\multicolumn{3}{c}{1991}&
\multicolumn{3}{c}{$1994$}&\multicolumn{3}{c}{}\\
\cmidrule(lr){2-4}\cmidrule(lr){5-7}
\cmidrule(lr){8-10}
&NBR&\multicolumn{2}{c}{HNBR}&NBR&\multicolumn{2}{c}{HNBR}
&NBR&\multicolumn{2}{c}{HNBR}&\\
\cmidrule(lr){3-4}\cmidrule(lr){6-7}
\cmidrule(lr){9-10}
&	&$\mu(x)$	&$k(x)$	&  &$\mu(x)$	& $k(x)$	&  & $\mu(x)$	& $k(x)$ &	&	&\\
\cmidrule(lr){1-10}
Female	&0	&0	&0	&0	&0	&0	&0	&0	&0	&	&	&	\\
Age	&-0.015	&-0.014	&-0.012	&-0.022	&-0.019	&-0.003	&-0.005	&-0.004	&-0.011	&	&	&	\\
Hsat	&-0.191	&-0.187	&-0.015	&-0.112	&-0.132	&-0.049	&-0.226	&-0.224	&-0.06	&	&	&	\\
Handdum	&0	&0	&0	&0	&0	&0	&0	&0	&0	&	&	&	\\
Handper	&0.011	&0.009	&0.006	&0.014	&0.013	&0	&0.007	&0.008	&0.004	&	&	&	\\
Hhninc	&0	&0	&0	&0	&0	&0	&0	&0	&0	&	&	&	\\
Hhkids	&0	&0	&0	&0	&0	&0	&0	&0	&0	&	&	&	\\
Educ	&-0.016	&-0.023	&-0.002	&-0.074	&-0.068	&0	&-0.064	&-0.069	&0	&	&	&	\\
Others	&0	&0	&0	&0	&0	&0	&0	&0	&0	&	&	&	\\
\cmidrule(lr){1-10}
FE	&1.144	&\multicolumn{2}{c}{0.912}	&1.007	&\multicolumn{2}{c}{0.787}	&0.713	&\multicolumn{2}{c}{0.58}	&	&	&	\\
\cmidrule(lr){1-10}
	\end{tabular}
\end{table}

\section{Conclusion and future study} 
We study the high-dimensional heterogeneous overdispersed count data via negative binomial regression models, and propose a double $\ell_1$-regularized method for simultaneous variable selection and dispersion estimation.
Under the restricted eigenvalue conditions, we prove the oracle inequalities with lasso estimators of two partial regression coefficients for the first time, using concentration inequalities of empirical processes. 
Furthermore, we derive the consistency and convergence rate for the estimators, which are the theoretical guarantees for further statistical inference. 
Simulation studies and a real example from the German  health  care  demand data indicate that the proposed method works satisfactorily.

In this work, we assume that the responses are independent. In the time-series data, however, the NB responses are temporal
dependent \citep{yang2021law}. Thus,  weak dependence conditions, including $\rho$-mixing, $m$-dependent types, could be considered in the future. The consequences of the main result would motivate further study for statistical inference, such as testing heterogeneous
$$
H_{0}: \theta^{(2)} =0  \text { vs. } H_{1}: \theta^{(2)} \ne 0 .
$$
The issues concerning the hypothesis testing are via the debiased Lasso estimator; see \cite{shi2019linear} and references therein. Another possible study is the false discovery rate (FDR) control, which aims to identify some small number of statistically significantly nonzero results after getting the sparse penalized estimation of HNBR; see \cite{xie2021aggregating,cui2021directional}.

 \appendix
 \section{Proofs}
  \renewcommand{\appendixname}{Appendix~\Alph{section}}
The first step is giving the property of the loss function. From mathematical analysis, we prefer bounded things to unlimited things.
Denote $\partial_j$ is the first partial differentiation with respect to $s_j$. The bounded aspect for $y$ and $s$ gives a nice property for the loss function $\gamma (s , y) = \gamma (z \theta, y)$.

\begin{lemma}\label{lem1}
    We have
    \begin{equation*}
        \partial_1 \gamma(s, y) = - \frac{y e^{s_2}}{e^{s_1} + e^{s_2}} + \frac{e^{s_1 + s_2}}{e^{s_1} + e^{s_2}}, \qquad \partial_2 \gamma(s, y) = \nu(s, y) + \frac{y e^{s_2}}{e^{s_1} + e^{s_2}}
    \end{equation*}
    where $\nu (s, y) = e^{s_2} \big( \psi(e^{s_2}) - \psi(y + e^{s_2}) \big) + e^{s_2} \log \big( 1 + e^{s_1 - s_2}\big) - \frac{e^{s_1 + s_2}}{e^{s_1} + e^{s_2}}$ satisfying
    \begin{equation*}
        \sup_{s \in \mathcal{S}, \, y \in \mathcal{Y}} | \nu (s, y)| \leq F_1, \qquad \sup_{s \neq t \in \mathcal{S}, \, y \in \mathcal{Y}} \frac{|\nu (s, y) - \nu (t, y)|}{\| s - t\|_{\infty}} \leq F_2,
    \end{equation*}
    with $F_1 =  M_{y, n} \big( 1 + e^{- m_{s, n}}\big) + e^{M_{s, n}} + \frac{e^{2 M_{s, n}}}{2 e^{m_{s, n}}}$, and
    \begin{equation*}
        \begin{aligned}
            F_2 & = 2 \left[ \left|e^{M_{s, n}} \left( 1 + \log(M_{y, n} + e^{M_{s, n}}) - \frac{1}{2(M_{y, n} + e^{M_{s, n}})}\right) \right| \vee \left| 1 - m_{s, n} e^{m_{s, n}}\right| \right]  \\
            & \qquad \qquad \qquad + \left[ \frac{e^{2 M_{s, n}}}{e^{m_{s, n}}} + \frac{2 e^{2 M_{s, n}}}{e^{m_{s, n}} + e^{M_{s, n}}}\right] + \frac{3}{2} e^{M_{s, n}}.
        \end{aligned}
    \end{equation*}
\end{lemma}
\begin{proof}
    We will use the properties of the psi function, the logarithmic derivative of the gamma function, to prove this lemma. Write $\psi (x) = \Gamma^{\prime}(x) / \Gamma(x)$. For any $s \in \mathcal{S}, \ y \in \mathcal{Y}$, using the Binet ’s formula (see p.18 of \cite{bateman1953higher})
    \begin{equation*}
        \psi(x) = \log x - \int_0^{\infty} \varphi(t) e^{- t x} \, dt,
    \end{equation*}
    where $\varphi(t)=1 /(1-e^{-t})-1 / t$ is  strictly increasing on $(0, \infty)$, it gives
    \begin{equation*}
        0 < \psi^{\prime} (x) = \frac{1}{x} + \int_0^{\infty} t \varphi(t) e^{- t x} \, dt  \leq \frac{1}{x} + \int_0^{\infty} t e^{- t x} \, dt = \frac{1}{x} + \frac{1}{x^2}.
    \end{equation*}
    and $y \geq 0$, we have
    \begin{equation*}
        \begin{aligned}
            |\nu (s, y)| & = \left| e^{s_2} \big( \psi (e^{s_2})  - \psi (y + e^{s_2}) \big) +  e^{s_2} \log \big( 1 + e^{s_1 - s_2}\big) - \frac{e^{s_1 + s_2}}{e^{s_1} + e^{s_2}} \right| \\
            & \leq e^{s_2} y \left( \frac{1}{e^{s_2}} + \frac{1}{e^{s_2}}\right) + e^{s_2} e^{s_1 - s_2} + \frac{e^{s_1 + s_2}}{e^{s_1} + e^{s_2}} \leq M_{y, n} \big( 1 + e^{- m_{s, n}}\big) + e^{M_{s, n}} + \frac{e^{2 M_{s, n}}}{2 e^{m_{s, n}}}.
        \end{aligned}
    \end{equation*}
    Then the first inequality in the lemma has been verified.
    On the other hand, by using the fact that (see (2.2) in \cite{alzer1997some})
    \begin{equation*}
        \frac{1}{2 x}<\log (x)-\psi(x)<\frac{1}{x}, \qquad x > 0,
    \end{equation*}
    for the function $f_1(x) = e^x \psi (e^x)$ and $f_2(x) = e^x \psi(y + e^x)$,
    \begin{gather*}
        f_1^{\prime}(x) = e^x \psi (e^x) + e^{2x} \psi^{\prime} (e^x) \leq e^x \left( x - \frac{1}{2e^x}\right) + e^{2x} \left( \frac{1}{e^x} + \frac{1}{e^{2x}}\right) = (x + 1) e^x + \frac{1}{2},\\
         f_1^{\prime}(x) = e^x \psi (e^x) - e^{2x} \psi^{\prime} (e^x) \geq e^x \psi (e^x ) \geq x e^x - 1, \\
         f_2^{\prime} (x) = e^x \psi(y + e^x) + e^{2x} \psi^{\prime}(y + e^x) \leq e^x \left( 1 + \log (y + e^x) - \frac{1}{2(y + e^x)}\right) + 1, \\
         f_2^{\prime} (x) \geq e^x \left( \log (y + e^x) - \frac{1}{y + e^x}\right)  \geq x e^x -1,
    \end{gather*}
    and for any $s \neq t \in \mathcal{S}$, $y \in \mathcal{Y}$, we conclude that
    \begin{equation*}
        |\psi(e^{s_2}) e^{s_2} - \psi(e^{t_2}) e^{t_2}| = |f_1(s_2) - f_1(t_2)| \leq \left( \left| (M_{s, n} + 1) e^{M_{s, n}}  + {1} / {2}\right| \vee \left| 1 - m_{s, n} e^{m_{s, n}}\right|  \right) \| s - t\|_{\infty},
    \end{equation*}
    and
    \begin{equation*}
        \begin{aligned}
            & |\psi(y + e^{s_2}) e^{s_2} - \psi(y + e^{t_2}) e^{t_2}| = |f_2(s_2) - f_2(t_2)| \\
            & \leq \left[ \left|e^{M_{s, n}} \left( 1 + \log(M_{y, n} + e^{M_{s, n}}) - \frac{1}{2(M_{y, n} + e^{M_{s, n}})}\right) \right| \vee \left| 1 - m_{s, n} e^{m_{s, n}}\right| \right] \| s - t\|_{\infty}.
        \end{aligned}
    \end{equation*}
    Besides, using the median value theorem again, we also have
    \begin{equation*}
        \begin{aligned}
            & \left| e^{s_2} \log(1 + e^{s_1 - s_2}) - e^{t_2} \log(1 + e^{t_1 - t_2}) \right| \\
            \leq & \log (1 + e^{s_1 - s_2}) |e^{s_2} - e^{t_2}| + e^{t_2} \left| \log(1 + e^{s_1 - s_2}) - \log(1 + e^{t_1 - t_2}) \right| \\
            \leq & e^{2M_{s, n} - m_{s, n}} |s_2 - t_2| + e^{M_{s, n}}\frac{1}{1 + e^{- (M_{s, n} - m_{s, n})}} |(s_1 - s_2) - (t_1 - t_2)| \\
            \leq & \left[ \frac{e^{2 M_{s, n}}}{e^{m_{s, n}}} + \frac{2 e^{2 M_{s, n}}}{e^{m_{s, n}} + e^{M_{s, n}}}\right] \| s - t\|_{\infty},
        \end{aligned}
    \end{equation*}
    and
    \begin{equation*}
        \begin{aligned}
            \left| \frac{e^{s_1 + s_2}}{e^{s_1} + e^{s_2}} - \frac{e^{t_1 + t_2}}{e^{t_1} + e^{t_2}} \right| & \leq e^{s_1} \left| \frac{1}{1 + e^{s_1 - s_2}} - \frac{1}{1 + e^{t_1 - t_2}}\right| + \frac{1}{1 + e^{t_1 - t_2}} |e^{s_1} - e^{t_1} | \\
            & \leq e^{M_{s, n}} \frac{1}{4} |(s_1 - s_2) - (t_1 - t_2)| + 1 \times e^{M_{s, n}} |s_1 - t_1|  \leq \frac{3}{2} e^{M_{s, n}} \| s - t \|_{\infty},
        \end{aligned}
    \end{equation*}
    where the fact used is that $f_3(x) = 1 / (1 + e^x)$ satisfies $| f_3^{\prime} (x)| = 1 / (e^x + e^{-x} + 2) \leq 1 / 4$. Since
    \begin{equation*}
        \begin{aligned}
            |\partial_2 \gamma(s, y) - \partial_2 \gamma(t, y)| & \leq \left| \psi(e^{s_2}) e^{s_2} - \psi (e^{t_2}) e^{t_2}\right| + \left| \psi(y + e^{s_2}) e^{s_2} - \psi (y + e^{t_2}) e^{t_2}\right| \\
            & \qquad + \left| e^{s_2} \log(1 + e^{s_1 - s_2}) - e^{t_2} \log(1 + e^{t_1 - t_2}) \right| + \left| \frac{e^{s_1 + s_2}}{e^{s_1} + e^{s_2}} - \frac{e^{t_1 + t_2}}{e^{t_1} + e^{t_2}} \right|,
        \end{aligned}
    \end{equation*}
    we can conclude the second inequality in the lemma.
\end{proof}

The Lemma separates the partial derivative of $\gamma$ into two parts: the first part is the linear about the response variable $y$ (say $-y e^{s_2} / (e^{s_1} + e^{s_2})$, $e^{s_1 + s_2} / (e^{s_1} + e^{s_2})$, and $y e^{s_2} / (e^{s_1} + e^{s_2})$), the second part is other complicated functions (not linear function) about $y$. The first part is relatively easy to analyze since the following concentration inequality gives a measure of dispersion about the weighted summation of negative binomial variables. This concentration inequality is a special case for the weighted summation of a series of random variables, which can be proved by sub-exponential concentration results in Proposition 4.2 in \cite{zhang2020concentration}.

\begin{lemma}\label{lem2}
    Suppose $\{ Y_i\}_{i = 1}^n$ are independently distributed as $\operatorname{NB}(\mu_i, k_i)$. Then for any non-random weights $w = (w_1, \cdots, w_n)^{\top} \in \mathbb{R}^n$ independent with $\{Y_i\}_{i = 1}^n$ and $t \geq 0$,
    \begin{equation*}
        \pr \bigg(\Big|\sum_{i=1}^{n} w_i (Y_{i} - \E Y_i) \Big| \geq t\bigg) \leq 2 \exp \left\{-\frac{1}{4}\left(\frac{t^{2}}{2 \sum_{i=1}^{n} w_i^2 a^2(\mu_i, k_i)} \wedge \frac{t}{\displaystyle\max _{1 \leq i \leq n}|w_i|a(\mu_i, k_i) }\right)\right\},
    \end{equation*}
    where $q_{i}:=\frac{\mu_{i}}{k_{i}+\mu_{i}} \in(0,1)$ and $a (\mu, k) := \left[ \log \frac{1 - (1 - q) / \sqrt[k]{2}}{q}\right]^{-1} + \frac{\mu}{\log 2}.$
\end{lemma}

\begin{proof}
    We will use the sub-exponential norm. The moment generating function (MGF) for $Y_i$ is
    \begin{equation*}
         \E e^{s Y_i} = \left( \frac{1 - q_i}{1 - q_i e^s}\right)^{k_i}.
    \end{equation*}
    Then by letting $\E \exp (|Y_i| / t) \leq 2$, we have
    \begin{equation*}
        2 \geq \E \exp (|Y_i| / t) = \E \exp (Y_i / t) = \left( \frac{1 - q_i}{1 - q_i e^{1 / t}}\right)^{k_i},
    \end{equation*}
    which implies the sub-exponential norm for $Y_i$ is
    \begin{equation*}
        \| Y_i\|_{\psi_1} = \inf \{t>0: \E \exp (|Y_i| / t) \leq 2\} = \left[ \log \frac{1 - (1 - q_i) / \sqrt[k_i]{2}}{q_i}\right]^{-1}.
    \end{equation*}
    Using the definition of $a_i$, from Proposition 4.2 in \cite{zhang2020concentration}, we can immediately obtain the result in the Lemma.
\end{proof}

{
{It should be noted that $a_i = a(\mu_i, k_i)$ naturally has a lower and upper bound for any $i \in [n]$ since $\mu_i$ and $k_i$ are both bounded between $e^{m_{s, n}}$ and $e^{M_{s, n}}.$}}
\vspace{1ex}

Note that $Y_i$ is an unbounded random variable; the next step is to find a probabilistic bound for $M_{y, n} = \max_{i \in [n]} Y_i$. We will cite an important lemma for this type of problem. We say a distribution $P_{\gamma}$ is strongly discrete log-concave with $\gamma > 0$ if its density is strongly midpoint log-convex with the same $\gamma > 0$.

\begin{lemma}\label{log_concave.lem}[Concentration for strongly log-concave discrete distributions]
    Let $P_{\gamma}$ be any strongly log-concave discrete distribution index by $\gamma>0$ on $\mathbb{Z}^{n}$. Then for any function $f: \mathbb{R}^{n} \rightarrow \mathbb{R}$ that is L-Lipschitz with respect to Euclidean norm, we have for $X \sim P_{\gamma}$,
    \begin{equation*}
        \pr_{P_{\gamma}} \big( |f(X) - \E f(X)| \geq t \big) \leq 2 \exp \left\{ - \frac{\gamma t^2}{4 L^2}\right\}
    \end{equation*}
    for any $t > 0$.
\end{lemma}

{
{
\begin{lemma}\label{unbounded_y.lem}
    The maximal of the response $M_{y, n} = \max_{i \in [n]} Y_i$ has the concentration
    \begin{equation*}
        \pr \left\{ M_{y, n} - \left( 2 \max_{i \in [n]} \left[ a(\mu_i, k_i) - \frac{\mu_i}{\log 2} \right] \big[ \log (2n) + \sqrt{2 \log (2n)}\big] + \max_{i \in [n]} \mu_i \right)  >  t\right\} \leq  e^{- \gamma t^2 / 4}
    \end{equation*}
    for any $t > 0$.
\end{lemma}

\begin{proof}
    For the upper bound of expectation, we first note that $Y_i - \E Y_i \sim \operatorname{subE}(2 \| Y_i\|_{\psi_1})$ with $ \| Y_i\|_{\psi_1}$ has calculated in Lemma \ref{lem2}, then we have $Y_i - \E Y_i \sim \operatorname{sub \Gamma}(4 \| Y_i\|_{\psi_1}^2, 2 \| Y_i\|_{\psi_1})$ by Example 5.3 in \cite{zhang2020concentration}, which further gives
    \begin{equation*}
        \begin{aligned}
            \E M_{y, n} & \leq \E \max_{i \in [n]} (Y_i - \E Y_i ) + \max_{i \in [n]} \E Y_i \\
            & \leq \big( 2 \cdot \max_{i \in [n]} 4 \| Y_i\|_{\psi_1}^2 \cdot \log (2n) \big)^{\frac{1}{2}} + \max_{i \in [n]} 2 \| Y_i\|_{\psi_1} \cdot \log (2n)  + \max_{i \in [n]} \mu_i \\
            & = 2 \max_{i \in [n]} \| Y_i\|_{\psi_1} \big[ \log (2n) + \sqrt{2 \log (2n)}\big] + \max_{i \in [n]} \mu_i,
        \end{aligned}
    \end{equation*}
    where the second $\leq $ is by Corollary 7.3 in \cite{zhang2020concentration} and the bound in the lemma comes from the explicit expression in Lemma \ref{lem2}.
    \vspace{1ex}
    
    \noindent By implementing Lemma \ref{log_concave.lem}, the remaining we need to do is verifying $\mathbb{Y} := (Y_1, \ldots, Y_n)^{\top} \in \mathbb{Z}^n$ belongs to some strongly log-concave discrete distribution $P_{\gamma}$ with the specifying $\gamma > 0$ after we take $f : (x_1, \ldots, x_n) \mapsto \max_{i \in [n]} x_i$ which is $1$-Lipschitz. By the definition, the derivative of log-density for $y := (y_1, \ldots, y_n)^{\top}$ is
    \begin{equation*}
        \psi^{\prime} (y_i) := \left. \frac{\partial \log p(y)}{\partial y} \right|_{y_i} = \log \frac{\Gamma (k_i + y_i)}{\Gamma(1 + y_i)} - y_i \log (k_i + \mu_i),
    \end{equation*}
    then the Taylor expansion gives
    \begin{gather*}
        \psi(y)=\psi\left(\left[\frac{1}{2} x+\frac{1}{2} y\right]\right)+\frac{1}{2} \psi^{\prime}\left(\left[\frac{1}{2} x+\frac{1}{2} y\right]\right)(y-x)+\frac{1}{8}(y-x)^{2} \psi^{\prime \prime}\left(a_{1}\right),\\
        \psi(x)=\psi\left(\left\lfloor\frac{1}{2} x+\frac{1}{2} y\right\rfloor\right)+\frac{1}{2} \psi^{\prime}\left(\left\lfloor\frac{1}{2} x+\frac{1}{2} y\right\rfloor\right)(x-y)+\frac{1}{8}(y-x)^{2} \psi^{\prime \prime}\left(a_{2}\right)
    \end{gather*}
    where $a_{1}=t_{1} y+\left(1-t_{1}\right)(x+y) / 2$,  $a_{2}=t_{2} y+\left(1-t_{1}\right)(x+y) / 2$ with $t_1, t_2 \in [0, 1]$. Define the difference function
    \begin{equation*}
        \Delta(x, y) := \frac{x-y}{4}\left[\psi^{\prime}\left(\left[\frac{1}{2} x+\frac{1}{2} y\right]\right)-\psi^{\prime}\left(\left[\frac{1}{2} x+\frac{1}{2} y\right]\right)\right]+\frac{\psi^{\prime \prime}\left(a_{1}\right)+\psi^{\prime \prime}\left(a_{2}\right)}{16}(y-x)^{2},
    \end{equation*}
    the Taylor expression above immediately implies 
    \begin{equation*}
        \Delta(x, y) \geq |x-y|^{2}\left\{\frac{\psi^{\prime \prime}\left(a_{1}\right)+\psi^{\prime \prime}\left(a_{2}\right)}{16}-\sup _{x \neq y ; x, y \in \mathbb{Z}^{n}} \frac{\left|\left[\psi^{\prime}(\lfloor(x+y) / 2\rfloor)-\psi^{\prime}(\lceil(x+y) / 2\rceil)\right]\right|}{4|x-y|}\right\}.
    \end{equation*}
    Let
    \begin{equation*}
        \begin{aligned}
            C_{\psi} &:=\sup _{x \neq y ; x, y \in \mathbb{Z}^{n}} \frac{\left|\left[\psi^{\prime}(\lfloor(x+y) / 2\rfloor)-\psi^{\prime}(\lceil(x+y) / 2\rceil)\right]\right|}{4|x-y|} \\
            & = \sup _{x \neq y ; x, y \in \mathbb{Z}^{n}} \left|\left[\log \frac{\Gamma(k_i +\lfloor(x+y) / 2\rfloor) \Gamma(\lceil(x+y) / 2\rceil+1)}{\Gamma(k_i +\lceil(x+y) / 2\rceil) \Gamma(\lfloor(x+y) / 2\rfloor+1)}-\frac{(\lfloor(x+y) / 2\rfloor-\lceil(x+y) / 2\rceil)}{\log ^{-1}\left(k_i + \mu_{i}\right)}\right]\right| / 4|x-y|,
        \end{aligned}
    \end{equation*}
    and it is not hard to see $C_{\psi} \approx \frac{\left|\log \left(k_i +\mu_{i}\right)\right|}{4} $ or 0. Besides,
    \begin{equation*}
        \begin{aligned}
            \psi^{\prime \prime}(y) & :=\left.\frac{\partial^{2} \log p(y)}{\partial y^{2}}\right|_{y=y_{i}}=\frac{d}{d y_{i}} \log \frac{\Gamma\left(\theta+y_{i}\right)}{\Gamma\left(y_{i}+1\right)} \\
            & = \sum_{m = 1}^{\infty}\left(\frac{1}{m+1}-\frac{1}{m+ k_i +y_{i}}\right)-\sum_{m=1}^{\infty}\left(\frac{1}{m+1}-\frac{1}{m+y_{i}+1}\right) \\
            & = \sum_{m = 1}^{\infty}\left(\frac{1}{m + y_{i} + 1}-\frac{1}{m + k_i +y_{i}}\right) \geq \inf _{y_{i} \in \mathbb{Z}} \sum_{m = 1}^{\infty}\left(\frac{1}{m + y_{i} + 1}-\frac{1}{m + k_i +y_{i}}\right) = C_{\psi^{\prime \prime}}.
        \end{aligned}
    \end{equation*}
    Now, we have got
    \begin{equation*}
        \Delta(x, y) \geq|x-y|^{2}\left\{\frac{\psi^{\prime \prime}\left(a_{1}\right)+\psi^{\prime \prime}\left(a_{2}\right)}{16}-C_{\psi}\right\} \geq|x-y|^{2}\left(\frac{C_{\psi^{\prime \prime}}}{8}-C_{\psi}\right)
    \end{equation*}
    which gives $\gamma=: \frac{C_{\psi^{\prime \prime}}}{8}-C_{\psi}>0$ from the strong log-concave assumption for $\mathbb{Y}$, if $C_{\psi} \approx \frac{|\log (k_i +\mu_{i}) |}{4}$ is small. Hence, we can conclude from Lemma \ref{log_concave.lem} and the upper bound of $\E M_{y, n}$
    \begin{equation*}
        \begin{aligned}
            & \pr \left\{ M_{y, n} - \left( 2 \max_{i \in [n]} \left[ a(\mu_i, k_i) - \frac{\mu_i}{\log 2} \right] \big[ \log (2n) + \sqrt{2 \log (2n)}\big] + \max_{i \in [n]} \mu_i \right)  >  t\right\} \\
            \leq & \pr (M_{y, n} - \E M_{y, n} > t ) \leq  e^{- \gamma t^2 / 4}
        \end{aligned}
    \end{equation*} 
    which is exact the result in the lemma. 
\end{proof}
}}

\begin{rem}
    For $M_{y, n}$, it is distributed as sub-Gumbel, which is rarely studied by research. Another way to deal with using the extreme value theory (EVT) technique, we note that for any $t \in \mathbb{R}$
    \begin{equation*}
        \begin{aligned}
            \pr (M_{y, n} - \E M_{y, n} > t) & = 1 - \prod_{i = 1}^n \pr \big( Y_i \leq t + \E M_{y, n} \big) \\
            & = 1 -  \prod_{i = 1}^n \bigg[ 1 - \exp \Big\{ - \frac{1}{4} \Big( \frac{(t + \E M_{y, n} - \mu_i)^2}{2 a_i^2} \wedge \frac{t + \E M_{y, n} - \mu_i}{a_i} \Big) \Big\} \bigg].
        \end{aligned}
    \end{equation*}
    If $Y_i$ is i.i.d, then in asymptotic sense,
    \begin{equation*}
        \begin{aligned}
            \pr (M_{y, n} - \E M_{y, n} > t) & = 1 - \Big[ 1 - \pr \big( Y_1 > t + \E M_{y, n} \big) \big]^n \\
            & \sim 1 - \exp \left\{ - n \pr \big( Y_1 > t + \E M_{y, n} \big) \right\} + o(1) \\
            & \sim 1 - \exp \left\{ - n \exp \Big(  - \frac{1}{4} \frac{(t + \E M_{y, n} - \mu_1)^2}{2 a_1^2} \wedge \frac{t + \E M_{y, n} - \mu_1}{a_1} \Big) \right\}.
        \end{aligned}
    \end{equation*}
    Unfortunately, this technique cannot be used in the above lemma since: (i) we need non-asymptotic version inequality instead of a vague expression with $n \rightarrow \infty$; (ii) $\{ Y_i \}$ is not an i.i.d series, and then EVT theory will not be easily used in this particular setting. Hence we adopt a discrete technique which has been use in \cite{moriguchi2020discrete} and fully illustrated in \cite{zhang2022elastic}.
\end{rem}

The stochastic Lipschitz conditions are established by using the properties of $\partial_1 \gamma(s, y)$ and $\partial_2 \gamma(s, y)$. As we said before, they are divided into two parts. The linear parts in them can be solved by the concentration inequality for NB variables given in Lemma \ref{lem2}, but the non-linear part $\nu(s, y)$ needs some more advanced tools regarding the empirical process. They are given as the following lemmas.
\begin{lemma}[The (3.12) in \cite{sen2018gentle}]\label{ex.lem1}
    Suppose $X_1(\omega), \cdots, X_n(\omega) \in \mathbb{R}$ are zero-mean independent stochastic processes indexed by $\omega \in \Omega$. If there exist $M_0$ and $S_0$ satisfying $|X_i(\omega)| \leq M_0$ and $\sum_{i = 1}^n \var \big( X_i(\omega)\big) \leq S_0^2$ for all $\omega \in \Omega$. Denote $S_n = \sup_{\omega \in \Omega} \big| \sum_{i = 1}^n  X_i(\omega)\big|$, then for any $t > 0$,
    \begin{equation*}
        \pr \big( S_n \geq 2 \E S_n + S_0 \sqrt{2t} + 4 M_0 t\big) \leq e^{-t}.
    \end{equation*}
\end{lemma}
A map $\phi: \mathbb{R} \rightarrow \mathbb{R}$ is called a contraction if $|\phi(s)-\phi(t)| \leq|s-t|$ for all $s, t \in \mathbb{R}$. And in the following lemmas $\varepsilon_1, \cdots, \varepsilon_n$ are always i.i.d. Rademacher variables.
\begin{lemma}[Theorem 2.2 in \cite{chi2010stochastic}]\label{ex.lem2}
    Let $\mathcal{T} \subseteq V^n$ be a bounded set and $f_1, \cdots, f_n$ be functions $V \rightarrow \mathbb{R}$ such that $f_i$ is $(M_i, \ell_{\infty})$-Lipschitz with $f_i(0) = 0$. For $j = 1, \cdots, k \in \mathbb{N}$, let $T_j = \{ (t_{1j}, \cdots, t_{nj}) : (t_1, \cdots, t_n) \in \mathcal{T}\} \subseteq \mathbb{R}^n$. Then
    \begin{equation*}
        \E \sup_{t \in \mathcal{T}} \bigg| \sum_{i = 1}^n \varepsilon_i f_i(t_i)\bigg| \leq \beta_k \sum_{j = 1}^k \E \sup_{s \in T_j} \bigg| \sum_{i = 1}^n \varepsilon_i M_i s_i \bigg|,
    \end{equation*}
    where $\beta_k$ is a universal constant that can be set no greater than $3^k + 3^{k - 1} - 2^k$.
\end{lemma}
\begin{lemma}[Theorem 4.12 in \cite{ledoux2013probability}]\label{ex.lem3}
    Let $F: \mathbb{R}_{+} \rightarrow \mathbb{R}_{+}$ be convex and increasing. Let further $\phi_{i}: \mathbb{R} \rightarrow \mathbb{R}, i \leq n$ be contractions such that $\phi_{i}(0)=0$. Then, for any bounded subset $\mathcal{T}$ in $\mathbb{R}^{n}$,
    \begin{equation*}
        \E F \bigg( \frac{1}{2} \sup_{\mathcal{T}} \Big| \sum_{i = 1}^n \varepsilon_i \phi_i(t_i) \Big|\bigg) \leq \E F \bigg( \sup_{\mathcal{T}} \Big| \sum_{i = 1}^n \varepsilon_i t_i \Big|\bigg).
    \end{equation*}
\end{lemma}
\begin{lemma}[Lemma 5.2 in \cite{massart2000some}]\label{ex.lem4}
    Let $\mathcal{A}$ be some finite subset of $\mathbb{R}^{n}$, let $R=\sup _{a \in \mathcal{A}}\left[\sum_{i=1}^{n} a_{i}^{2}\right]^{1 / 2}$, then
    \begin{equation*}
        \E \bigg[\sup _{a \in \mathcal{A}} \sum_{i=1}^{n} \varepsilon_{i} a_{i}\bigg] \leq R \sqrt{2 \log \big( \operatorname{card}(\mathcal{A})\big)}.
    \end{equation*}
\end{lemma}
With the assistance of these powerful tools, we can establish the stochastic Lipschitz condition as follows, which is one of the most important points in this article for establishing the oracle inequality of the $\ell_2$ distance between the estimated value $\widehat{\theta}$ and the real value $\theta^*$.
\vspace{2ex}

\textbf{The proof of Corollary \ref{thm1}.}
\begin{proof}
    Denote $c_i = Z_i \theta^*$. For $\theta \in \Theta$, denote $t_i = Z_i(\theta - \theta^*) = Z_i \theta - c_i$. We also define the map $\bar{\pi}_j : (x_1, \cdots, x_p)^{\top} \mapsto (x_1, \cdots, x_j, 0, \cdots, 0)^{\top}$ and the function
    \begin{equation*}
        \varphi_{ij} (s) = \left\{ \begin{array}{ll}
            \displaystyle\frac{\gamma(c_i + \bar{\pi}_j s, Y_i) - \gamma(c_i + \bar{\pi}_{j - 1} s, Y_i)}{s_j} - \partial_j \gamma(c_i, Y_i), & \text{ if } s_j \neq 0; \\
            \partial_j \gamma(c_i + \bar{\pi}_{j - 1} s, Y_i) - \partial_j \gamma(c_i, Y_i), & \text{ if } s_j = 0.
        \end{array} \right.
    \end{equation*}
    Thus, $\varphi_{ij} : \mathbb{R}^2 \rightarrow \mathbb{R}$ is a real value function for $i = 1, \cdots, n, \, j = 1, 2$. Then it is easy to check that
    \begin{equation*}
        \gamma(Z_i \theta, Y_i) - \gamma(Z_i\theta^*, Y_i) = \sum_{j = 1}^2 \big( \partial_j \gamma(c_i, Y_i) + \varphi_{ij} (t_i) \big) t_{ij},
    \end{equation*}
    and $n \pn \big( \gamma(\theta) - \gamma(\theta^*) \big) = \sum_{i = 1}^n \sum_{j = 1}^2 \big( \partial_j \gamma(c_i, Y_i) + \varphi_{ij} (t_i) \big) X_i^{\top} \big( \theta^{(j)} - \theta^{ * (j)} \big)$
    in turn. It gives
    \begin{equation*}
        \begin{aligned}
            \sqrt{n} \gn \big( \gamma(\theta) - \gamma(\theta^*)\big) = & \sum_{i = 1}^n \sum_{j = 1}^2 \big( \partial_j \gamma(c_i, Y_i) - \E \partial_j \gamma(c_i, Y_i) \big) X_i^{\top} \big( \theta^{(j)} - \theta^{ * (j)} \big) \\
            & + \sum_{i = 1}^n \sum_{j = 1}^2 \big( \varphi_{ij}(t_i) - \E \varphi_{ij}(t_i)  \big) X_i^{\top} \big( \theta^{(j)} - \theta^{ * (j)} \big).
        \end{aligned}
    \end{equation*}

    First, we would like to give the explicit formula for $\varphi_{i1}$ and obtain an upper bound as well as a Lipschitz parameter for $\varphi_{i2}$. Denote $h_i (\cdot) = \gamma(\cdot, Y_i)$, then
    \begin{equation*}
        \varphi_{ij} (s) = \int_0^1 \big( \partial_j h_i (c_i + \bar{\pi}_{j - 1} s + s_j u e_j) - \partial_j h_i(c_i)\big) \, du,
    \end{equation*}
    where $e_j$ is the $j$-th basis vector of $\mathbb{R}^2$. Hence, for $j = 1$,
    \begin{equation*}
        \begin{aligned}
            \varphi_{i1}(s) & = - Y_i \int_0^1 \left[ \frac{e^{s_2}}{e^{s_1 + u} + e^{s_2}} - \frac{e^{s_2}}{e^{s_1} + e^{s_2}}\right] \, du + \int_0^1 \left[ \frac{e^{s_1 + s_2 + u}}{e^{s_1 + u} + e^{s_2}} - \frac{e^{s_1 + s_2}}{e^{s_1} + e^{s_2}}\right] \, du \\
            & = \left[ \log \frac{e^{s_1 + 1} + e^{s_2}}{e^{s_1} + e^{s_2}} - \frac{e^{s_1}}{e^{s_1} + e^{s_2}}\right] Y_i + C_1 (s),
        \end{aligned}
    \end{equation*}
    in which $C_1(s)$ is a function only related to $s$ and free of $Y$ and the index $i$. Using Lemma \ref{lem1}, for $j = 2$, write $F_3 = F_1 + M_{y, n}$,
    \begin{equation*}
        |\varphi_{i2} (s)| \leq \int_0^1 \big| \partial_2 h_i (c_i + \bar{\pi}_1 s + s_2 u e_2) - \partial_2 h_i(u)\big| \, du \leq 2 F_3,
    \end{equation*}
    and
    \begin{equation*}
        \begin{aligned}
            \big| \varphi_{i2} (s)  - \varphi_{i2} (t)\big| & \leq \int_0^1 \big| \partial_2 h_i (c_i + \bar{\pi}_1 s + s_2 u e_2) - \partial_2 h_i (c_i + \bar{\pi}_1 t + t_2 u e_2) \big| \, du \\
            & \leq \int_0^1 F_2 \| \bar{\pi}_1 (s - t) + (s_2 - t_2) u e_2\|_{\infty} \, du \leq F_2 \| s - t \|_{\infty}.
        \end{aligned}
    \end{equation*}
    This implies $\varphi_{i 2}$ is $(F_2, \ell_{\infty})$ Lipschitz. In particular, letting $s = Z_i (\theta - \theta^*)$ and $t = 0$,
    \begin{equation*}
        \left| \varphi_{i 2} \big(Z_i (\theta - \theta^*)\big)\right| \leq \| Z_i (\theta - \theta^*)\|_{\infty} \leq F_2 M_x D_{\Theta}.
    \end{equation*}
    Hence, we obtain an upper bound for $\varphi_{i 2}$ that
    \begin{equation}\label{varphi.i2.bound}
        \left| \varphi_{i 2} \big(Z_i (\theta - \theta^*)\big)\right| \leq 2 F_3 \vee F_2 M_x D_{\Theta} := M_1
    \end{equation}

    Now, for $k = 1, \cdots, p$, define
    \begin{equation*}
        \xi_{ik} (\theta) := \big( \varphi_{i2}(t_i) - \E \varphi_{i2}(t_i)  \big) X_{i k}, \qquad S_k = \sup_{\theta \in \Theta} \left| \sum_{i = 1}^n \xi_{ik} (\theta) \right|.
    \end{equation*}
    Then we can approach the final conclusion in the theorem by
    \begin{equation}\label{thm1.approach}
        \begin{aligned}
            \sup_{\theta \in \Theta / \{\theta^*\}} \left| \frac{\sqrt{n} \gn \big( \gamma (\theta) - \gamma (\theta^*)\big)}{\|\theta - \theta^*\|_1} \right| & \leq \max_{1 \leq k \leq p} \left| \sum_{i = 1}^n \big( \partial_1 \gamma(c_i, Y_i) - \E \partial_1 \gamma (c_i, Y_i) \big) X_{i k} \right| \\
            & \qquad + \sup_{\theta \in \Theta / \{\theta^*\}} \max_{1 \leq k \leq p} \left| \sum_{i = 1}^n \big( \varphi_{i1}(t_i) - \E \varphi_{i1}(t_i) \big) X_{i k} \right| \\
            & \qquad + \max_{1 \leq k \leq p} \left| \sum_{i = 1}^n \frac{e^{c_{i2}}}{e^{c_{i1}} + e^{c_{i2}}} \big( Y_i - \E Y_i \big) X_{ik}\right| \\
            & \qquad + \max_{1 \leq k \leq p} \left| \sum_{i = 1}^n \big( \nu(c_i, Y_i) - \E \nu (c_i, Y_i) \big) X_{i k} \right|  +  \sup_{\theta \in \Theta / \{\theta^*\}} \max_{1 \leq k \leq p} S_k.
        \end{aligned}
    \end{equation}
    We will tickle with (\ref{thm1.approach}) term by term.

    \vspace{2ex}
    \emph{(i). the first three term in (\ref{thm1.approach}):}
    \vspace{2ex}

    We will use concentration inequality to deal with these terms. For any $1 \leq k \leq p$ and $t \geq 0$, by Lemma \ref{lem2} and Cauchy-Schwartz inequality,
    \begin{equation*}
        \begin{aligned}
            & \pr \bigg( \Big| \sum_{i = 1}^n \big( \partial_1 \gamma(c_i, Y_i) - \E \partial_1 \gamma (c_i, Y_i) \big) X_{i k} \Big| \geq t \bigg) = \pr \bigg( \Big| \frac{e^{c_{i2}}}{e^{c_{i1}} + e^{c_{i2}}} X_{ik} \big( Y_i - \E Y_i\big) \Big| \geq t \bigg) \\
            & \leq 2 \exp \left\{-\frac{1}{4}\left(\frac{t^{2}}{2 \sum_{i=1}^{n} (w_i^{(1)})^2 X_{ik}^2 a_i^2} \wedge \frac{t}{\displaystyle\max _{1 \leq i \leq n}|w_i^{(1)} X_{ik}|a_i }\right)\right\} \\
            & \leq 2 \exp \left\{-\frac{1}{4}\left(\frac{t^{2}}{2  \sqrt{\sum_{i=1}^{n} (w_i^{(1)})^4 a_i^4} \displaystyle\max _{1 \leq k \leq p}\sqrt{\begin{matrix}\sum_{i=1}^{n} X_{ik}^4 \end{matrix}}}\wedge \frac{t}{M_x \displaystyle\max _{1 \leq i \leq n} |w_i^{(1)}| a_i  }\right)\right\},
        \end{aligned}
    \end{equation*}
    where $w_i^{(1)} = e^{c_{i2}} / (e^{c_{i1}} + e^{c_{i2}})$ and $a_i = a(\mu_i, k_i)$ is defined in Lemma \ref{lem2}, they are both determined and free of $\theta$ and the index $k$. Hence,
    \begin{equation*}
        \begin{aligned}
            & \pr \bigg( \max_{1 \leq k \leq p} \Big| \sum_{i = 1}^n \big( \partial_1 \gamma(c_i, Y_i) - \E \partial_1 \gamma (c_i, Y_i) \big) X_{i k} \Big| \geq t \bigg) \\
            & \qquad \qquad \leq 2 p \exp \left\{-\frac{1}{4}\left(\frac{t^{2}}{ 2 \sqrt{\sum_{i=1}^{n} (w_i^{(1)})^4 a_i^4} \displaystyle\max _{1 \leq k \leq p}\sqrt{\begin{matrix}\sum_{i=1}^{n} X_{ik}^4 \end{matrix}}}\wedge \frac{t}{M_x \displaystyle\max _{1 \leq i \leq n} |w_i^{(1)}| a_i  }\right)\right\}.
        \end{aligned}
    \end{equation*}
    By letting the right side of the above display be $q_1 \in (0, 1)$, we can get
    \begin{equation*}
        \begin{aligned}
            & \pr \bigg( \max_{1 \leq k \leq p} \Big| \sum_{i = 1}^n \big( \partial_1 \gamma(c_i, Y_i) - \E \partial_1 \gamma (c_i, Y_i) \big) X_{i k} \Big|  \\
            & \quad \geq 2 \sqrt{2 \big( \sum_{i=1}^{n} (w_i^{(1)})^4 a_i^4 \big)^{1/2} \big( \max_{1 \leq k \leq p} \sum_{i=1}^{n} X_{ik}^4 \big)^{1/2} \log (2 p / q_1)} \vee 4 M_x \max_{1 \leq i \leq n} |w_i^{(1)}| a_i \log (2p / q_1)\Bigg) \leq q_1.
        \end{aligned}
    \end{equation*}
    Exactly the same, we can get for any $q_3 \in (0, 1)$, regarding to the third term,
    \begin{equation*}
        \begin{aligned}
            & \pr \Bigg( \max_{1 \leq k \leq p} \Big| \sum_{i = 1}^n \frac{e^{c_{i2}}}{e^{c_{i1}} + e^{c_{i2}}}\big( Y_i - \E Y_i \big) X_{i k} \Big|  \\
            & \quad \geq 2 \sqrt{2 \big( \sum_{i=1}^{n} (w_i^{(1)})^4 a_i^4 \big)^{1/2} \big( \max_{1 \leq k \leq p} \sum_{i=1}^{n} X_{ik}^4 \big)^{1/2} \log (2 p / q_3)} \vee 4 M_x \max_{1 \leq i \leq n} |w_i^{(1)}| a_i \log (2p / q_3)\Bigg) \leq q_3.
        \end{aligned}
    \end{equation*}

    The situation is slightly different for the second term. Indeed,
    \begin{equation*}
        \begin{aligned}
            \pr  \left( \Big| \sum_{i = 1}^n \big( \varphi_{i1}(t_i) - \E \varphi_{i1} (t_i)\big) X_{ik} \Big| \geq t \right)  & = \pr \left( \left| \sum_{i = 1}^n \left[ \log \frac{e^{t_{i1} + 1} + e^{t_{i2}}}{e^{t_{i1}} + e^{t_{i2}}} - \frac{e^{t_{i1}}}{e^{t_{i1}} + e^{t_{i2}}}\right] X_{ik} (Y_i - \E Y_i) \right| \geq t\right) \\
            & := \pr \left( \Big| \sum_{i = 1}^n w_i^{(2)}(\theta) X_{ik} \big( Y_i - \E Y_i \big)\Big| \geq t \right)
        \end{aligned}
    \end{equation*}
    Since $t_i$ is a function of $\theta$, so as the weights $w_i^{(2)}(\theta)$, we cannot use the exact same method as previous. However, since $\Theta$ is convex, we have $\{t_i\}_{i = 1}^n \subseteq \mathcal{S}$. Then it only needs to note that,
    \begin{equation*}
        \big|w_i^{(2)}(\theta)\big| = \left|\log \frac{e^{t_{i1} + 1} + e^{t_{i2}}}{e^{t_{i1}} + e^{t_{i2}}} - \frac{e^{t_{i1}}}{e^{t_{i1}} + e^{t_{i2}}}\right| \leq \log \frac{e + e^{M_{s, n} - m_{s, n}}}{1 + e^{m_{s, n} - M_{s, n}}} + \frac{1}{1 + e^{m_{s, n} - M_{s, n}}} := w^{(2)},
    \end{equation*}
    which gives
    \begin{equation*}
        \begin{aligned}
            & \pr \Bigg( \max_{1 \leq k \leq p} \Big| \sum_{i = 1}^n \big( \varphi_{i1}(t_i) - \E \varphi_{i1} (t_i)\big) X_{ik} \Big| \\
            & \geq 2 \sqrt{2 \sqrt{n} w^{(2)2}\big( \sum_{i=1}^{n} a_i^4 \big)^{1/2} \big( \max_{1 \leq k \leq p} \sum_{i=1}^{n} X_{ik}^4 \big)^{1/2} \log (2 p / q_2)} \vee 4 M_x w^{(2)}\max_{1 \leq i \leq n} | a_i| \log (2p / q_2)\Bigg) \leq q_2.
        \end{aligned}
    \end{equation*}
    for any $\theta \in \Theta$ and $q_2 \in (0, 1)$.

    \vspace{2ex}
    \emph{(ii). the forth term in (\ref{thm1.approach}):}
    \vspace{2ex}

    From Lemma \ref{lem1}, we know that $|\nu (c_i, Y_i)| \leq F_1$. Thus, simply by Hoeffding inequality (see Corollary 2.1 (b) in \cite{zhang2020concentration}), for any $t \geq 0$ and $1 \leq k \leq p$,
    \begin{equation*}
        \begin{aligned}
            \pr \bigg( \Big| \sum_{i = 1}^n \big( \nu(c_i, Y_i) - \E \nu (c_i, Y_i) \big) X_{i k} \Big| \geq t \bigg) & \leq 2 \exp \left\{ - \frac{t^2}{2F_1^2 \sum_{i = 1}^n X_{ik}^2}\right\} \leq 2 \exp \left\{ - \frac{t^2}{2F_1^2 \displaystyle\max_{1 \leq k \leq p} \begin{matrix} \sum_{i = 1}^n X_{ik}^2 \end{matrix} } \right\}.
        \end{aligned}
    \end{equation*}
    For arbitrary $q_4 \in (0, 1)$, let $t = F_1 \sqrt{2 \log (2 p / q_4) \max_{1 \leq k \leq p} \sum_{i = 1}^n X_{ik}^2}$, we obtain
    \begin{equation*}
        \pr \bigg( \max_{1 \leq k \leq p} \Big| \sum_{i = 1}^n \big( \nu(c_i, Y_i) - \E \nu(c_i, Y_i) \big) X_{i k} \Big| \geq F_1 \sqrt{2 \log (2 p / q_4) \max_{1 \leq k \leq p} \sum_{i = 1}^n X_{ik}^2} \bigg) \leq q_4.
    \end{equation*}

    \emph{(iii). the last term in (\ref{thm1.approach}):}
    \vspace{2ex}

    For any $i = 1, \cdots, n$ and $k = 1, \cdots, p$, by (\ref{varphi.i2.bound}), $|\xi_{ik} (\theta)| \leq 2 M_1 M_x := M_0$. And for any $\theta \in \Theta$, (\ref{varphi.i2.bound}) also implies
    \begin{equation*}
        \sum_{i = 1}^n \var \big( \xi_{ik} (\theta) \big) = \sum_{i = 1}^n \E \big( \varphi_{i2} (t_i) X_{ik} \big)^2 \leq A_1^2 \sum_{i = 1}^n X_{ik}^2 \leq A_1^2 \max_{1 \leq k \leq p} \sum_{i = 1}^n X_{ik}^2 := S_0^2
    \end{equation*}
    Therefore, from Lemma \ref{ex.lem1}, it follows that
    \begin{equation}\label{E.Sk}
        \pr \left( S_k \geq 2 \E S_k + S_0 \sqrt{2 t} + 4 M_0 t\right) \leq e^{-t}.
    \end{equation}
    So the last task is giving an upper bound for $\E S_k$. Note that $\E \xi_{ik} (\theta) = 0$, by symmetrization,
    \begin{equation*}
        \begin{aligned}
            \E S_ k  = \E \sup_{\theta \in \Theta} \bigg| \sum_{i = 1}^n \big( \varphi_{i2} (t_i) - \E \varphi_{i2} (t_i) \big) X_{ik}\bigg|  \leq 2 \E \sup_{\theta \in \Theta} \bigg| \sum_{i = 1}^n \varepsilon_i \varphi_{i2} (t_i)  X_{ik}\bigg| = 2 \E \sup_{t \in \mathcal{T}} \bigg| \sum_{i = 1}^n \varepsilon_i \varphi_{i2} (t_i)  X_{ik}\bigg|,
        \end{aligned}
    \end{equation*}
    where $\mathcal{T} = \{ t_i = Z_i(\theta - \theta^*): \theta \in \Theta, i = 1, \cdots, n\}$, and $\varepsilon_1, \cdots, \varepsilon_n$ are i.i.d. Rademacher variables independent of $Y_1, \cdots, Y_n$. Here using the fact $\varphi_{i2} (\cdot) X_{ik}$ is $(M_x F_2, \ell_{\infty})$-Lipschitz and Lemma \ref{ex.lem2}, \ref{ex.lem3}, and \ref{ex.lem4},
    \begin{equation*}
        \begin{aligned}
             \E \sup_{t \in \mathcal{T}} \bigg| \sum_{i = 1}^n \varepsilon_i \varphi_{i2} (t_i)  X_{ik}\bigg| & \leq 8 M_x F_2 \sum_{j = 1}^2 \E \sup_{t \in \mathcal{T}} \bigg| \varepsilon_i t_{ij}\bigg| = 8 M_x F_2 \sum_{j = 1}^2 \E \sup_{\theta \in \Theta} \bigg| \sum_{i = 1}^n \varepsilon_i X_i^{\top} \big(\theta^{(j)} - \theta^{* (j)} \big) \bigg| \\
             & \leq 16 M_x F_2 D_{\Theta} \E \max_{1 \leq k \leq p} \bigg| \sum_{i = 1}^n \varepsilon_i X_{ik}\bigg|\leq 16 \sqrt{2 \log p } M_x F_2 D_{\Theta} \sqrt{\displaystyle\max_{1 \leq k \leq p} \sum_{i = 1}^n X_{ik}^2}.
        \end{aligned}
    \end{equation*}
    Then by (\ref{E.Sk}),
    \begin{equation*}
        \pr \Bigg( S_k \geq 32 \sqrt{2 \log p } M_x F_2 D_{\Theta} \sqrt{\max_{1 \leq k \leq p} \sum_{i = 1}^n X_{ik}^2} + M_1 \sqrt{ 2t \max_{1 \leq k \leq p} \sum_{i = 1}^n X_{ik}^2 } + 8 M_1 M_x t\Bigg) \leq e^{-t}.
    \end{equation*}
    Note that the right side of the inequality is free of $\theta$, let $t = \log (p / q_5)$ in the above inequality, and use the same technique as previous, we get the uniform bound for it. And the Theorem is proved by letting $q_2 = q_3 = q_1$, $q_4 = q_2$, $q_5 = q_3$, and $|w_i^{(1)}| \leq w^{(1)}$.
\end{proof}
\vspace{2ex}

\textbf{The lower bound of the likelihood-based divergence}

Recall the standard steps for establishing oracle inequality for lasso estimator is (see \cite{xiao2020oracle} for example):
\begin{itemize}
    \item[I.] To avoid the ill behavior of Hessian, propose the restricted eigenvalue condition or other analogous conditions about the design matrix.

    \item[II.] Find the tuning parameter based on the high-probability event, i.e., the KKT conditions.

    \item[III.] According to some restricted eigenvalue assumptions and tuning parameter selection, derive the oracle inequalities via the definition of the Lasso optimality and the minimizer under unknown expected risk function and some basic inequalities. There are three sub-steps:

    \begin{itemize}
        \item[(i)]  Under the KKT conditions, show that the error vector $\widehat{\theta} - \theta^*$ is in some restricted set with structure sparsity, and check that $\widehat{\theta} - \theta^*$ is in a big compact set;

        \item[(ii)] Show that likelihood-based divergence of $\widehat{\theta}$ and $\theta^*$ can be lower bounded by some quadratic distance between $\widehat{\theta}$ and $\theta^*$;

        \item[(iii)] By some elementary inequalities and (ii), show that $\| \widehat{\theta} - \theta^* \|_1$ is in a smaller compact set with radius of optimal rate (proportional to $\lambda$).
    \end{itemize}
\end{itemize}

Under our approach, the KKT condition with a high probability is replaced by the stochastic Lipschitz condition, while other steps should remain the same. For most models belonging the canonical exponential family, the step III.(ii) is quite trivial, see Lemma 1 in \cite{abramovich2016model} for example. Nonetheless, it is worthy to note that the our loss function is not in the canonical exponential family, so there is no extended discussion about the lower bound of likelihood-based divergence of $\widehat{\theta}$ and $\theta^*$ in our setting. We will use the following theorem to clarify this thing.

\begin{theorem}\label{thm2}
    Suppose the condition is the same as that in Theorem \ref{thm1}. Denote the true parameter for $Y_i$ is $\mu^*$ and $k^*$. If $\{Z_i \theta\}_{i = 1, \cdots, n, \theta \in \Theta} \subseteq \mathcal{S} \cap \{ s \in \mathbb{R}^2 : 2 s_1 + (1 + s_2 (1 - k^*) k^{* \mu^*})\mu^* \leq \frac{s_1 + \mu^*}{2 s_2^2}\}$ and $\mu^* \geq 1$, then
    \begin{equation*}
        \E \gamma (Z_i \theta, Y_i) - \E \gamma (Z_i \theta^*, Y_i) \geq C_{\gamma} \|Z_i(\theta  - \theta^*) \|_2^2,
    \end{equation*}
    where $C_{\gamma}$ is a positive constant and its exact definition is in the proof.
\end{theorem}

\begin{proof}
    For simplicity, we drop the index $i$. By the definition and the notation in Theorem \ref{thm1},
    \begin{equation*}
        \begin{aligned}
            \E \gamma (Z \theta, Y) - \E \gamma (Z \theta^*, Y) & = D_{\operatorname{KL}} (s, c),
        \end{aligned}
    \end{equation*}
    where $D_{\operatorname{KL}}$ is the Kullback-Leibler divergence from the $Y_i$'s density $f (y \, | \, Z \theta)$ to $f (y \, | \, Z \theta^*)$, i.e.
    \begin{equation*}
        D_{\operatorname{KL}}(s, c) := \int f(y \, | \, c) \log \frac{f (y \, | \, c)}{f (y \, | \, s)} \, dy.
    \end{equation*}
    Due to the identification of the negative binomial distribution, we have $D_{\operatorname{KL}}(s, c) \geq 0$ with equality if and only if $s = c$. Using the Taylor theorem,
    \begin{equation*}
        \begin{aligned}
            D_{\operatorname{KL}}(s, c) & = D_{\operatorname{KL}}(c, c) + \left. \frac{\partial}{\partial s} D_{\operatorname{KL}}(s, c) \right|_{s = c} + \frac{1}{2} (s - c)^{\top} \left[ \frac{\partial^2}{\partial s \partial s^{\top}} D_{\operatorname{KL}}(s, c) \right]_{s = c + \rho (s - c)} (s - c) \\
            & = \frac{1}{2} (s - c)^{\top} \left[ \frac{\partial^2}{\partial s \partial s^{\top}} D_{\operatorname{KL}}(s, c) \right]_{s = c + \rho (s - c)} (s - c) \\
            & \geq \frac{1}{2} \inf_{\rho \in [0, 1]} \lambda_{min} \left[ \frac{\partial^2}{\partial s \partial s^{\top}} D_{\operatorname{KL}}(s, c) \right]_{s = c + \rho (s - c)} \|s - c\|_2^2
        \end{aligned}
    \end{equation*}
    where $\rho \in [0, 1]$ and $\lambda_{min} (\mathrm{M})$ is the smallest eigenvalue of the matrix $\mathrm{M}$. So it is enough to show that $\left[ \frac{\partial^2}{\partial s \partial s^{\top}} D_{\operatorname{KL}}(s, c) \right]_{s = c + \rho (s - c)}$ is strictly positive define for any $\rho \in [0, 1]$. First, calculate directly,
    \begin{equation*}
         \begin{aligned}
             \frac{\partial^2}{\partial s \partial s^{\top}} D_{\operatorname{KL}}(s, c) & = \int f(y \, | \, c) \left[ \frac{\partial^2}{\partial s \partial s^{\top}} \gamma(s, y) \right] \, dy \\
             & = \int f(y \, | \, c) \left[ \begin{array}{cc}
                \displaystyle\frac{e^{s_1 + s_2}}{(e^{s_1} + e^{s_2})^2} (e^{s_2} + y) & \displaystyle\frac{e^{s_1 + s_2}}{(e^{s_1} + e^{s_2})^2} (e^{s_1} - y)  \\
                 \displaystyle\frac{e^{s_1 + s_2}}{(e^{s_1} + e^{s_2})^2} (e^{s_1} - y) & \displaystyle \partial_2 \nu(s, y) + \frac{e^{s_1 + s_2}}{(e^{s_1} + e^{s_2})^2} y
             \end{array} \right] \, dy \\
             & =: \left[ \begin{array}{cc}
                a_{11} + b  & a_{12} - b \\
                a_{21} - b & a_{22} + b
             \end{array} \right],
         \end{aligned}
    \end{equation*}
    where $a_{11} = \frac{e^{s_1 + 2s_2}}{(e^{s_1} + e^{s_2})^2}, \qquad a_{12} = \frac{e^{2s_1 + s_2}}{(e^{s_1} + e^{s_2})^2}, \qquad b = \frac{e^{s_1 + 2s_2}}{(e^{s_1} + e^{s_2})^2} \E Y,$
    and
    \begin{equation*}
        \begin{aligned}
            a_{22} = \E \partial_2 v(s, Y) & = e^{s_2} \left[ \psi (e^{s_2}) + e^{s_2} \psi^{\prime} (e^{s_2}) + \log (1 + e^{s_1 - s_2}) - \frac{e^{s_1}}{e^{s_1} + e^{s_2}} - \Big( \frac{e^{s_1}}{e^{s_1} + e^{s_2}} \Big)^2 \right] \\
            & \qquad \quad - e^{s_2} \left[ \E \psi (Y + e^{s_2}) + e^{s_2} \E \psi^{\prime} (Y + e^{s_2}) \right].
        \end{aligned}
    \end{equation*}
    For a $2 \times 2$ matrix $\mathrm{M}$, it is strictly positive define if and only if $\operatorname{tr} (\mathrm{M}) > 0$ and $\operatorname{det} (\mathrm{M}) > 0$. Denote $\mu = e^{s_1}, k = e^{s_2}$, and $\mu^* = e^{c_1} , k^* = e^{c_2}$ are true parameters for $Y$. Then
    \begin{equation}\label{trace}
        \begin{aligned}
            \operatorname{tr} \left[ \frac{\partial^2}{\partial s \partial s^{\top}} D_{\operatorname{KL}}(s, c) \right] & = \frac{\mu k^2}{(\mu + k)^2} + 2 \frac{\mu k^2}{(\mu + k)^2} \mu^* - k \left[ \frac{\mu}{\mu + k} + \Big( \frac{\mu}{\mu + k} \Big)^2 \right] \\
            & \qquad + k \left[ \log (1 + \mu / k) + \big(\psi(k) - \E \psi(Y + k) \big) + k \big( \psi^{\prime}(k) - \E \psi^{\prime}(Y + k)\big) \right] \\
            & = \frac{2(\mu^* - 1)\mu k^2}{(\mu + k)^2} + k \big[ \log(1 + \mu / k) + g_1(k) + k g_2(k) \big] \\
            & \geq k \big[ \log(1 + \mu / k) + g_1(k) + k g_2(k) \big].
        \end{aligned}
    \end{equation}

    Now, we are going to deal with $g_1(k) = \psi(k) - \E \psi(Y + k)$ and $g_2(k) = \psi^{\prime}(k) - \E \psi^{\prime}(Y + k)$. For $\psi(x)$,
    \begin{equation*}
        0 > \psi^{\prime \prime} (x) = - \frac{1}{x^2} - \int_0^{\infty} t^2 \varphi(t) e^{- tx} \, dt \geq - \frac{1}{x^2} - \frac{2}{x^3}.
    \end{equation*}
    Therefore, $\psi(\cdot)$ is concave. Using Jensen inequality and median value theorem
    \begin{equation*}
        \begin{aligned}
            g_1(k) = \psi(k) - \E \psi (Y + k) \geq \psi(k) - \psi \big( \E Y + k \big) \geq - \left( \frac{1}{k} + \frac{1}{k^2} \right) \E Y = - \mu^* \left( \frac{1}{k} + \frac{1}{k^2} \right).
        \end{aligned}
    \end{equation*}
    Similarly, for $g_2(k)$, by using the fact that $\E (1 / Y) = (1 - k^*) k^{* \mu^*} \mu^*$ and the assumption,
    \begin{equation*}
        \begin{aligned}
            g_2(k) & = \E \left[ \psi^{\prime} (k) - \psi^{\prime} (Y + k) \right] \geq \E \left[ Y \left( \frac{1}{(\xi(Y) + k)^2} + \frac{2}{(\xi(Y) + k)^3}\right)\right] \\
            & \geq \E \left[ \frac{Y}{(Y + k)^2}\right] + 2 \E \left[ \frac{Y}{(Y + k)^3} \right] \geq \left[ \E \frac{(Y + k)^2}{Y}\right]^{-1} + 2 \left[ \E \frac{(Y + k)^3)}{Y}\right]^{-1} \\
            & = \frac{1}{ 2 k + (1 + k^2 (1 - k^*) k^{* \mu^*}) \mu^* } + \frac{2}{k^{*2} (k^* + \mu^*) / \mu^{* 2} + \mu^{* 2} + 3 k \mu^* + 3 k^2 + k^3 (1 - k^*) k^{* \mu^*} \mu^*} \\
            & \geq (\mu + \mu^*) \left( \frac{1}{2k^2} + \frac{1}{k^3}\right).
        \end{aligned}
    \end{equation*}
    where $\xi(Y)$ lies between $0$ and $Y$. The lower bounds for $g_1$ and $g_2$, together with the fact that $\log (1 + x) \geq x - x^2/2$ for $x \geq 0$, we conclude that $\operatorname{tr} \left[ \frac{\partial^2}{\partial s \partial s^{\top}} D_{\operatorname{KL}} (s, c )\right] > 0$. Similarly, we can also prove $\operatorname{det} \left[ \frac{\partial^2}{\partial s \partial s^{\top}} D_{\operatorname{KL}} (s, c )\right] > 0$, so the theorem holds.
\end{proof}

\textbf{The proof of Theorem \ref{thm3}}
\begin{proof}
    The proof follows the idea in \cite{bickel2009simultaneous}. First, by the definition of $\widehat{\theta}$,
    \begin{equation*}
        \begin{aligned}
            P \big( \gamma(\widehat{\theta}) - \gamma (\theta^*)\big) & \leq P \big( \gamma(\widehat{\theta}) - \gamma (\theta^*)\big) + \big( \pn \gamma(\theta^*) + \lambda \|\theta^* \|_{\omega, 1}\big) - \big( \pn \gamma(\widehat{\theta}) + \lambda \|\widehat{\theta} \|_{\omega, 1} \big) \\
            & \leq \frac{1}{\sqrt{n}} \gn \big( \gamma (\theta^*) - \gamma (\widehat{\theta})\big) + \lambda \big( \|\theta^*\|_{\omega, 1} - \|\widehat{\theta}\|_{\omega, 1}\big).
        \end{aligned}
    \end{equation*}
    From Theorem \ref{thm2}, we also have
    \begin{equation*}
        P \big( \gamma(\widehat{\theta}) - \gamma (\theta^*)\big) \geq \frac{C_{\gamma}}{n} \sum_{i = 1}^n \|Z_i (\widehat{\theta} - \theta^*) \|_2^2 = \frac{C_{\gamma}}{n} \sum_{j = 1}^2 \| \X (\widehat{\theta}^{(j)} - \theta^{* (j)})\|_2^2.
    \end{equation*}
    Then by Theorem \ref{thm1} and the definition of $\lambda$,
    \begin{equation*}
        \begin{aligned}
            C_{\gamma} \sum_{j = 1}^2 \| \X (\widehat{\theta}^{(j)} - \theta^{* (j)})\|_2^2 & \leq {\sqrt{n}} \gn \big( \gamma (\theta^*) - \gamma (\widehat{\theta})\big) + n \lambda \big( \|\theta^*\|_{\omega, 1} - \|\widehat{\theta}\|_{\omega, 1} \big) \\ 
            & \leq M_q \| \theta^* - \widehat{ \theta } \|_1 + (1 + 1 / a) M_q \big( \|\theta^*\|_{\omega, 1} - \|\widehat{\theta}\|_{\omega, 1}\big) \\
            & = M_q \sum_{j =  1}^2 \bigg[ \| \widehat{\theta}^{(j)} - \theta^{* (j)}\|_1 + (1 + 1/a) \omega_j \big( \|\theta^{* (j)}\|_1 - \|\widehat{\theta}^{(j)}\|_1\big) \bigg]
        \end{aligned}
    \end{equation*}
    holds with probability at least $1 - q$, where $a = (K - 1) / 2$. Now let $J_1, J_2 \subseteq \{1, \cdots, p\}$ be any sets with $J_j \supseteq \spt \big( \theta^{*(j)}\big)$. It is easy to check
    \begin{equation*}
        \begin{aligned}
            \| \widehat{\theta}^{(j)} - \theta^{* (j)}\|_1 & + (1 + 1/a) \omega_j \big( \|\theta^{* (j)}\|_1 - \|\widehat{\theta}^{(j)}\|_1\big) \\
            & = \| \widehat{\theta}_{J_j}^{(j)} - \theta^{*(j)}\|_1 + \|\widehat{\theta}_{J_j^c}^{(j)}\|_1 + (1 + 1/a) \omega_j \big( \|\theta^{*(j)}\|_1 - \|\widehat{\theta}_{J_j}^{(j)}\|_1 - \|\widehat{\theta}_{J_j^c}^{(j)}\|_1\big) \\
            & \leq (K / a) \| \widehat{\theta}_{J_j}^{(j)} - \theta^{* (j)}\|_1 - (1 / a) \| \widehat{\theta}_{J_j^c}^{(j)} \|_1.
        \end{aligned}
    \end{equation*}
    by the fact $\omega_j \in [0, 1]$. It gives that with probability at least $1 - q$,
    \begin{equation}\label{oracle1}
        \sum_{j = 1}^2 \| \X (\widehat{\theta}^{(j)} - \theta^{* (j)})\|_2^2 \leq \frac{M_q}{a C_{\gamma}} \sum_{j = 1}^2 \big( K \| \widehat{\theta}_{J_j}^{(j)} - \theta^{* (j)}\|_1 - \| \widehat{\theta}_{J_j^c}^{(j)} \|_1 \big).
    \end{equation}

    Let $A_1, A_2 \subseteq \{1, \cdots, p\}$ satisfying $\spt \big( \theta^{*(j)}\big) \subseteq A_j$ and $\operatorname{card} (A_j) = p_1$, and we also let $B_j$ be the union of $A_j$ and the indices of $p_1$ largest $\widehat{\theta}^{(j)}$. Then $A_j$ and $B_j$ also guarantee (\ref{oracle1}). And from Lemma \ref{ex.lem5}, they also give
    \begin{equation*}
        \| \widehat{\theta}^{(j)}_{B_j^c}\|_2^2 \leq p_1^{-1} \| \widehat{\theta}^{(j)}_{A_j^c}\|_1^2.
    \end{equation*}
    Besides, from the definition of $A_j$ and $B_j$, we know that $\| \widehat{\theta}^{(j)}_{A_j^c}\|_1 \geq \| \widehat{\theta}^{(j)}_{B_j^c}\|_1$ and $\|\widehat{\theta}^{(j)}_{A_j} - \theta^{*(j)}\|_1 \leq \|\widehat{\theta}^{(j)}_{B_j} - \theta^{*(j)}\|_1 $.

    Unlike the single lasso question, here we need to define $I := \{j = 1, 2 : K \| \widehat{\theta}_{A_j}^{(j)} - \theta^{* (j)}\|_1 \geq \| \widehat{\theta}_{A_j^c}^{(j)} \|_1 \}$, and consider $j \in I$ and $j \notin I$ separately. Obviously, $I \neq \emptyset$, or (\ref{oracle1}) cannot be holden. For $j \in I$, we have
    \begin{equation*}
        K \|\widehat{\theta}^{(j)}_{B_j} - \theta^{*(j)}\|_1 - \| \widehat{\theta}^{(j)}_{B_j^c}\|_1 \geq K \|\widehat{\theta}^{(j)}_{A_j} - \theta^{*(j)}\|_1 - \| \widehat{\theta}^{(j)}_{A_j^c}\|_1 \geq 0.
    \end{equation*}
    Then by the restricted eigenvalue condition,
    \begin{equation*}
        n \kappa^2 \| \widehat{\theta}_{J_j}^{(j)} - \theta^{*(j)}\|_2^2 \leq \| \X (\widehat{\theta}^{(j)} - \theta^{* (j)})\|_2^2
    \end{equation*}
    holds for $J_j = A_j$ or $J_j = B_j$. Note that from (\ref{oracle1}),
    \begin{equation*}
        \sum_{j \in I} \| \X (\widehat{\theta}^{(j)} - \theta^{*(j)} )\|_2^2 \leq \frac{M_q}{a C_{\gamma}} \sum_{j \in I} \big( \| \widehat{\theta}_{A_j}^{(j)} - \theta^{*(j)}\|_1 - \|\widehat{\theta}_{A_j^c}^{(j)} \|_1\big) \leq \frac{M_q}{a C_{\gamma}} \sum_{j \in I} \big( \| \widehat{\theta}_{B_j}^{(j)} - \theta^{*(j)}\|_1 - \|\widehat{\theta}_{B_j^c}^{(j)} \|_1\big),
    \end{equation*}
    then by Cauchy-Schwartz inequality,
    \begin{equation*}
        \begin{aligned}
            n \kappa^2 \sum_{j \in I} \| \widehat{\theta}_{A_j}^{(j)} - \theta^{* (j)}\|_2^2 & \leq \| \X (\widehat{\theta}_{A_j}^{(j)} - \theta^{* (j)})\|_2^2 \leq \frac{M_q K}{a C_{\gamma}} \sum_{j \in I} \|\widehat{\theta}_{A_j}^{(j)} - \theta^{*(j)} \|_1 \\
            & \leq \frac{M_q K \sqrt{p_1}}{a C_{\gamma}} \sum_{j \in I}  \|\widehat{\theta}_{A_j}^{(j)} - \theta^{*(j)} \|_2 \leq \frac{M_q K \sqrt{2 p_1}}{a C_{\gamma}} \left[ \sum_{j \in I}  \|\widehat{\theta}_{A_j}^{(j)} - \theta^{*(j)} \|_2^2 \right]^{1/2}.
        \end{aligned}
    \end{equation*}
    It gives
    \begin{equation*}
        \sum_{j \in I}  \|\widehat{\theta}_{A_j}^{(j)} - \theta^{*(j)} \|_2^2 \leq \frac{2 p_1 M_q^2 K^2}{a^2 \kappa^4 n^2 C_{\gamma}^2}, \qquad \sum_{j \in I}  \|\widehat{\theta}_{B_j}^{(j)} - \theta^{*(j)} \|_2^2 \leq \frac{4 p_1 M_q^2 K^2}{a^2 \kappa^4 n^2 C_{\gamma}^2},
    \end{equation*}
    where we use that fact $\operatorname{card}(B_j) = 2 p_1$. Furthermore, since
    \begin{equation*}
        \begin{aligned}
             \| \widehat{\theta}^{(j)}_{B_j^c}\|_2^2 \leq \sum_{j \in I} p_1^{-1} {\|\widehat{\theta}^{(j)}_{A_j^c}\|_1^2} \leq \frac{K^2}{p_1} \sum_{j \in I} {\|\widehat{\theta}_{A_j}^{(j)} - \theta^{*(j)} \|_1^2} \leq K^2 \sum_{j \in I} {\|\widehat{\theta}_{A_j}^{(j)} - \theta^{*(j)} \|_2^2},
        \end{aligned}
    \end{equation*}
    we can conclude that
    \begin{equation}\label{thm3.final1}
        \begin{aligned}
            \sum_{j \in I} \| \widehat{\theta}^{(j)} - \theta^{* (j)}\|_2^2 & = \sum_{j \in I} \big( \|\widehat{\theta}_{B_j}^{(j)} - \theta^{*(j)} \|_2^2 + \|\widehat{\theta}^{(j)}_{B_j^c}\|_2^2\big) \\
            & \leq \sum_{j \in I} \big( \|\widehat{\theta}_{B_j}^{(j)} - \theta^{*(j)} \|_2^2 + K^2 \|\widehat{\theta}_{A_j}^{(j)} - \theta^{*(j)} \|_2^2 \big) = \frac{2 p_1 M_q^2 (2 + K^2)  K^2}{a^2 \kappa^4 n^2 C_{\gamma}^2}.
        \end{aligned}
    \end{equation}

    Now we will tickle the situation that $j \notin I$. For $j \notin I$, $K \|\widehat{\theta}_{A_j}^{(j)} - \theta^{*(j)} \|_1 < \| \widehat{\theta}_{A_j^c}^{(j)}\|_1$. Again from (\ref{oracle1}), we have
    \begin{equation*}
        \sum_{j \notin I} \| \X (\widehat{\theta}^{(j)} - \theta^{* (j)})\|_2^2 \leq \frac{M_q K}{a C_{\gamma}} \sum_{j \in I} \|\widehat{\theta}_{A_j}^{(j)} - \theta^{*(j)} \|_1
    \end{equation*}
    and
    \begin{equation*}
        0 \leq \sum_{j \notin I} \big(\| \widehat{\theta}_{A_j^c}^{(j)}\|_1 - K \|\widehat{\theta}_{A_j}^{(j)} - \theta^{*(j)} \|_1 \big) \leq K \sum_{j \in I} \|\widehat{\theta}_{A_j}^{(j)} - \theta^{*(j)} \|_1.
    \end{equation*}
    Indeed, if the two inequalities above has the opposite direction, then for the first one, one can find that
    \begin{equation*}
        \sum_{j \in I} \| \X (\widehat{\theta}^{(j)} - \theta^{* (j)})\|_2^2 \leq \frac{M_q}{a C_{\gamma}} \left[ \sum_{j \notin I}\big( K \|\widehat{\theta}_{A_j}^{(j)} - \theta^{*(j)} \|_1 - \|\widehat{\theta}_{A_j^c}^{(j)}\|_1 \big) - \sum_{j \in I} \| \widehat{\theta}_{A_j^c}^{(j)}\|_1\right] < 0,
    \end{equation*}
    and
    \begin{equation*}
        \sum_{j = 1}^2 \| \X (\widehat{\theta}^{(j)} - \theta^{* (j)})\|_2^2 \leq - \frac{M_q}{a C_{\gamma}} \sum_{j \in I}  \| \widehat{\theta}_{A_j^c}^{(j)}\|_1 < 0.
    \end{equation*}

    Once again, by Cauchy-Schwartz inequality,
    \begin{equation*}
        \sum_{j \in I} \|\widehat{\theta}_{A_j}^{(j)} - \theta^{*(j)} \|_1 \leq \sqrt{p_1} \sum_{j \in I} \|\widehat{\theta}_{A_j}^{(j)} - \theta^{*(j)} \|_2 \leq \sqrt{2 p_1} \left[ \sum_{j \in I} \|\widehat{\theta}_{A_j}^{(j)} - \theta^{*(j)} \|_2^2 \right]^{1 / 2} \leq \frac{2 p_1 M_q K}{a \kappa^2 n C_{\gamma}}.
    \end{equation*}
    Denote $\Delta_j := \| \widehat{\theta}_{A_j^c}^{(j)}\|_1 - K \|\widehat{\theta}_{A_j}^{(j)} - \theta^{*(j)} \|_1$. Then for $j \notin J$, $\Delta_j > 0$, and
    \begin{equation*}
        \sum_{j \notin I} \Delta_j \leq K \sum_{j \in I}  \|\widehat{\theta}_{A_j}^{(j)} - \theta^{*(j)} \|_1 \leq \frac{2 p_1 M_q K^2}{a \kappa^2 n C_{\gamma}}.
    \end{equation*}
    For any $j \notin I$, define
    \begin{equation*}
        \widetilde{\theta}^{(j)} = \widehat{\theta}^{(j)} + \frac{\Delta_j}{p_1 K} \sum_{k \in A_j} \operatorname{sgn} \big(\widehat{\theta}_{k}^{(j)} - \theta_{k}^{* (j)}\big) e_k.
    \end{equation*}
    Then for $k \in A_j$,
    \begin{equation*}
        |\widetilde{\theta}_k^{(j)} - \theta_{k}^{* (j)}| = |\widehat{\theta}_k^{(j)} - \theta_{k}^{* (j)}| + \frac{\Delta_j}{p_1 K},
    \end{equation*}
    while for $k \notin I$, $\widetilde{\theta}^{(j)}_k = \widehat{\theta}^{(j)}_k$. Therefore,
    \begin{equation*}
        K \| \widetilde{\theta}_{A_j}^{(j)} - \theta^{*(j)}\|_1 = K \left[ \| \widehat{\theta}_{A_j}^{(j)} - \theta^{*(j)}\|_1 + \sum_{k \in A_j} \frac{\Delta_j}{p_1 K}\right] = \| \widehat{\theta}_{A_j^c}^{(j)} \|_1 = \| \widetilde{\theta}_{A_j^c}^{(j)} \|_1,
    \end{equation*}
    and consequently $\| \widetilde{\theta}_{B_j^c}^{(j)} \|_1 \leq K \|\widetilde{\theta}_{B_j}^{(j)} - \theta^{* (j)} \|_1$. Once again by the restricted eigenvalue condition,
    \begin{equation}\label{oracle2}
        \| \X (\widetilde{\theta}^{(j)} - \theta^{* (j)})\|_2^2 \geq n \kappa^2 \| \widetilde{\theta}_{B_j}^{(j)} - \theta^{* (j)}\|_2^2 \geq  n \kappa^2 \| \widetilde{\theta}_{A_j}^{(j)} - \theta^{* (j)}\|_2^2.
    \end{equation}

    On the other hand, note that for any $s, t \in \mathbb{R}^m$ inequality $\|s + t\|_2^2 \leq 2 (\|s\|_2^2 + \|t\|_2^2)$ and $\|s\|_2 \leq \|s \|_1 \leq \sqrt{m} \|s\|_2$ hold, we conclude
    \begin{equation}\label{oracle3}
        \begin{aligned}
            \sum_{j \notin I} \| \X (\widetilde{\theta}^{(j)} - \theta^{* (j)}) \|_2^2 & \leq 2 \sum_{j \notin I} \Big( \| \X (\widehat{\theta}^{(j)} - \theta^{* (j)}) \|_2^2 + \| \X (\widehat{\theta}^{(j)} - \widetilde{\theta}^{(j)}) \|_2^2 \Big) \\
            & \leq \frac{2 M_q K}{a C_{\gamma}} \sum_{j \in I} \|\widehat{\theta}_{A_j}^{(j)} - \theta^{ * (j)} \|_1 + 2 \sum_{j \notin I} \| \X (\widehat{\theta}^{(j)} - \widetilde{\theta}^{(j)}) \|_2^2 \\
            & \leq \frac{4 p_1 M_q^2 K^2}{n a^2 \kappa^2 C_{\gamma}^2} + 2 \sum_{j \notin I} \| \X (\widehat{\theta}^{(j)} - \widetilde{\theta}^{(j)}) \|_2^2.
        \end{aligned}
    \end{equation}
    Next, we will use the definition of the $p_1$-restricted isometry constant $\sigma_{\X, l}^2$. Since $\spt \big( \widetilde{\theta}^{(j)} - \widehat{\theta}^{(j)} \big) \leq \operatorname{card} (A_j) = p_1$, then
    \begin{equation*}
        \begin{aligned}
            \sum_{j \notin I} \| \X (\widehat{\theta}^{(j)} - \widetilde{\theta}^{(j)}) \|_2^2 & \leq \sigma_{\X, p_1}^2 \sum_{j \notin I} \| \widehat{\theta}^{(j)} - \widetilde{\theta}^{(j)} \|_2^2 \\
            & = \sigma_{\X, p_1}^2 \sum_{j \notin I} \sum_{k \in A_j} \Big( \frac{\Delta_j}{p_1 K}\Big)^2 = \frac{\sigma_{\X, p_1}^2}{p_1 K^2} \sum_{j \notin I} \Delta_j^2 \\
            & \leq \frac{\sigma_{\X, p_1}^2}{p_1 K^2} \Big( \sum_{j \notin I} \Delta_j \Big)^2 \leq \frac{4 p_1 \sigma_{\X, p_1}^2 K^2}{a^2 \kappa^4 n^2 C_{\gamma}^2}.
        \end{aligned}
    \end{equation*}
    The above inequality together with (\ref{oracle2}) and (\ref{oracle3}) gives
    \begin{equation*}
        \sum_{j \notin I} \| \widetilde{\theta}_{A_j}^{(j)} - \theta^{* (j)}\|_2^2 \leq \sum_{j \notin I} \| \widetilde{\theta}_{B_j}^{(j)} - \theta^{* (j)}\|_2^2 \leq \frac{4 p_1 (n \kappa^2 + 2 \sigma_{\X, p_1}^2) M_q^2 K^2}{a^2 C_{\gamma}^2 n^3 \kappa^6}.
    \end{equation*}
    Finally, since
    \begin{equation*}
        \| \widetilde{\theta}_{B_j^c}^{(j)}\|_2^2 \leq \| \widetilde{\theta}_{B_j^c}^{(j)}\|_1^2 \leq K^2  \| \widetilde{\theta}_{B_j}^{(j)} - \theta^{* (j)}\|_1^2 \leq 2 p_1 K^2 \| \widetilde{\theta}_{B_j}^{(j)} - \theta^{* (j)}\|_2^2,
    \end{equation*}
    we obtain that
    \begin{equation}\label{thm3.final2}
        \begin{aligned}
            \sum_{j \notin I} \| \widehat{\theta}^{(j)} - \theta^{* (j)}\|_2^2 & \leq \sum_{j \notin I} \| \widetilde{\theta}^{(j)} - \theta^{* (j)}\|_2^2 = \sum_{j \notin I} \big(  \| \widetilde{\theta}_{B_j}^{(j)} - \theta^{* (j)}\|_2^2 + \| \widetilde{\theta}_{B_j^c}^{(j)}\|_2^2 \big) \\
            & \leq (1 + 2p_1 K)  \sum_{j \notin I} \| \widetilde{\theta}_{B_j}^{(j)} - \theta^{* (j)}\|_2^2 \leq \frac{4 p_1 (1 + 2p_1 K) (n \kappa^2 + 2 \sigma_{\X, p_1}^2) M_q^2 K^2}{a^2 C_{\gamma}^2 n^3 \kappa^6}.
        \end{aligned}
    \end{equation}
    Combining (\ref{thm3.final1}) and (\ref{thm3.final2}), it is easy to see the remaining 
\end{proof}

\bibliographystyle{chicago}
\bibliography{reference}
\end{document}